\documentclass[a4paper,UKenglish,cleveref,autoref,thm-restate]{lipics-v2021}

\hideLIPIcs
\title{Price of Censorship: Censorship Resistance and Throughput under Rational Concurrent Proposers}

\titlerunning{Censorship Resistance and Throughput under Rational Concurrent Proposers}%

\usepackage{graphicx}%
\usepackage{amsmath, amsthm}
\usepackage{amssymb}
\usepackage{nameref}
\usepackage{float}
\usepackage{xspace}
\usepackage{algorithm}
\usepackage{algpseudocodex}
\usepackage{subcaption}

\newtheorem{property}{Property}

\renewcommand{\paragraph}[1]{\medskip\noindent\textbf{#1}}

\author{Aditya Saraf}{Cornell University, Ithaca, NY, USA}{as2777@cornell.edu}{0000-0003-4466-8144}{}

\author{Ioannis Kaklamanis}{Yale University, IC3, New Haven, CT, USA}{giannis.kaklamanis@yale.edu}{0009-0000-9916-1709}{}

\author{Sarisht Wadhwa}{Duke University, Durham, NC, USA}{sarisht.wadhwa@duke.edu}{0000-0003-4343-9868}{}

\author{Fatima Elsheimy}{Yale University, New Haven, CT, USA}{fatima.elsheimy@yale.edu}{0009-0006-2315-766X}{}

\authorrunning{A. Saraf, I. Kaklamanis, S. Wadhwa and F. Elsheimy }

\Copyright{Aditya Saraf, Ioannis Kaklamanis, Sarisht Wadhwa, and Fatima Elsheimy}



\category{}%

\relatedversiondetails{This is the full version for the paper accepted at AFT '26}{}%

\nolinenumbers%

\newcommand{\E}{\ensuremath{\mathbb{E}}}
\newcommand{\R}{\ensuremath{\mathbb{R}}}
\newcommand{\tx}{\ensuremath{\textit{tx}}\xspace}
\newcommand{\numtransactions}{\ensuremath{m}\xspace}

\newcommand{\adv}{\mathcal{A}}
\newcommand{\pvec}{\mathbf{p}}
\newcommand{\tvec}{\mathbf{t}}
\newcommand{\pmass}{\ell}

\newcommand{\censorprob}{s}
\newcommand{\equil}{\textsc{equilibrium}}
\newcommand{\lambdalow}{\underline{\lambda}}
\newcommand{\lambdahigh}{\overline{\lambda}}
\newcommand{\game}{G}
\newcommand{\tip}{t}
\newcommand{\fclamp}{\hat{f}}

\EventEditors{Aggelos Kiayias and Maria Kyropoulou}
\EventNoEds{2}
\EventLongTitle{8th Conference on Advances in Financial Technologies (AFT 2026)}
\EventShortTitle{AFT 2026}
\EventAcronym{AFT}
\EventYear{2026}
\EventDate{October 6--9, 2026}
\EventLocation{London, UK}
\EventLogo{}
\SeriesVolume{395}
\ArticleNo{11}

\supplementdetails[subcategory={Source code and data}]{Dataset}{https://zenodo.org/records/21831640}

\begin{document}

\maketitle

\begin{abstract}

   Censorship resistance is the defining advantage of blockchains over their centralized counterparts. Yet block proposers censor transactions for many reasons, from legal consequences to economic incentives. We study \emph{economically-incentivized} censorship, modeled by an adversary who bribes proposers to exclude a target transaction, and define the \emph{economic censorship resistance} (eCR) of a transaction as the adversary's expected cost of successful censorship divided by the user's expected payment for inclusion. Single-proposer systems are structurally weak by this measure: under a first-price auction the adversary need only match the user's bid, and fee burning pushes eCR to a few percent of what the user pays.

We therefore turn to multiple concurrent proposers (MCP), where block capacity is divided among $n$ proposers and the block is the union of their sub-blocks.While MCP can substantially increase the cost of censorship by requiring the adversary to bribe many proposers, it also introduces transaction duplication, reducing throughput. The resulting trade-off depends critically on the transaction fee mechanism (TFM), which determines how fees are shared among competing proposers.

    We create a game theoretic model where validators construct blocks from a shared mempool, subject to an adversary's bribery attempt. We provide an algorithm that solves for the mixed equilibrium of a given mempool, which is characterized by the probability of including each transaction. This algorithm works for a wide class of TFMs, and allows us to calculate the expected throughput and censorship resistance for any bid distribution. We then use simulations to show how the eCR and throughput vary as the number of proposers increases. We compare three TFMs, finding that the \textit{duplication-penalizing} TFM dominates the others across many settings. We also validate our findings with empirical Ethereum data.

    We find that MCP systems enjoy much higher eCR than single proposer systems. While burnt user fees can significantly lower eCR, eCR scales linearly with the number of validators, and thus strong eCR can be achieved even in the presence of significant burn. Overall, our results show that MCP systems are, with the correct TFM, significantly more robust to censorship.
\end{abstract}

\section{Introduction}
\label{sec:introduction}
Blockchains are decentralized across many blocks, but within a single block they are anything but: the current proposer is a monopolist with complete control over the block's contents. Decentralization is meant to give every transaction in the mempool a ``fair chance'' at inclusion, relative to its bid, unlike centralized systems, where the operator can simply refuse to process a payment or a post~\cite{eff_financial_censorship}. In practice, the single-slot monopoly undermines this promise. Wahrst\"atter et al.~\cite{WahrstaetterErnstbergerYaish2024BlockchainCensorship} measure that 46\% of post-merge Ethereum blocks were produced by actors censoring OFAC-sanctioned transactions, and as of early 2026, live dashboards report that roughly 35\% of Ethereum builders and half of all relays engage in some form of transaction censorship~\cite{censorship-website}. These measurements suggest that transaction censorship is a persistent feature of today's blockchain ecosystem.

Some of this censorship is legally motivated, and its cost is mainly latency: a transaction censored by regulated proposers waits until a validator in a different jurisdiction proposes. Economically motivated censorship is more corrosive, because the censor can profit from every slot of delay. Consider a participant in an on-chain auction, a DeFi liquidation, or an arbitrage race: excluding a competitor's transaction for even a few blocks can be worth orders of magnitude more than the fees that transaction pays. Since a bribe only needs to outbid the proposer's \emph{fee} from including the transaction, not the value the transaction protects, censorship is profitable precisely when value at stake far exceeds fees, which is the defining regime of maximal extractable value (MEV). An attacker can thus censor a competitor indefinitely by repeatedly offering each proposer a modest bribe.

In this paper, we focus on such \emph{economically incentivized censorship}, which we model as an adversary who publicly offers every proposer a bribe for excluding a target transaction. We quantify a system's resilience by its \emph{economic censorship resistance} (eCR): the ratio between the adversary's expected total expenditure for successful censorship and the fee the user expects to pay for inclusion. For example, if a user expects to pay 1 ETH, but an adversary must spend 5 ETH to censor the transaction, then the eCR is 5.

The eCR of existing designs is inherently limited. With a single proposer per block and a first-price auction transaction fee mechanism (TFM), the eCR is at most $1$: the adversary simply matches the user's bid. Fee burning makes matters strictly worse. On Ethereum, the proposer receives only the tip while the base fee is burned; in typical network conditions, tips are on the order of $0.5$--$2$ gwei against base fees of $10$--$50$ gwei, yielding $\text{eCR} \approx \frac{\text{tip}}{\text{base fee} + \text{tip}} \in [0.02, 0.1]$. An adversary can censor a transaction on today's Ethereum by paying a few percent of what the user pays.

\paragraph{Multiple concurrent proposers}
To achieve higher censorship resistance, we must break the proposer's monopoly by having multiple concurrent proposers (MCP). MCP designs range from research protocols such as MirBFT~\cite{stathakopoulou-mirbft-21}, to DAG-based systems deployed in production by Sui and Aptos~\cite{mystenlabs2022suiwhitepaper,babel_mysticeti_2024,aptos2025whitepaper}, to active proposals for Ethereum and Solana~\cite{mcp}.

We propose a simple MCP model that captures many of these systems: the total block capacity is divided among $n$ concurrent proposers, and the final block is the union of their sub-blocks. Some transactions may then be included by multiple proposers, which yields two complications. First, the TFM must dictate how a transaction's fee is split among the proposers.~\cite{garimidi2025transactionfeemechanismdesign} analyzes a TFM that pools all tips and divides them equally among proposers, regardless of who includes which transactions. ~\cite{Fox2023Auctions} solicits a large and a small tip from each user: a transaction's sole includer receives the large tip, while multiple includers each receive the small one. As the Ethereum calculation above suggests, this choice heavily affects censorship resistance. We mostly analyze an extreme version of the TFM of~\cite{Fox2023Auctions}, which we term the \emph{duplication-penalizing} TFM: it solicits a single tip and pays it to the includer only if the transaction is included exactly once. Our equilibrium analysis, however, holds for a large family of TFMs including all these examples.

The second complication is that multiple inclusions cause a loss of \textit{throughput}, which we define as the percentage of \textit{unique} transactions in a block. There is a natural tension between censorship resistance and throughput: doubling the number of proposers increases the amount the adversary must spend, but also creates more opportunities for duplication. The TFM again matters, as different TFMs discourage duplication to different degrees.

MCP systems are adopted precisely for these two advertised benefits: higher throughput, from parallel block production, and censorship resistance, from the fact that no single proposer can unilaterally exclude a transaction. But a design cannot simply maximize both. The redundancy that protects transactions manifests, in equilibrium, as duplication that consumes the very capacity parallelism creates, and where a system lands on this frontier depends directly on its transaction fee mechanism. This dependence frames the central question of our work:
\begin{quote}
    \emph{Which transaction fee mechanism best navigates the trade-off between censorship resistance and throughput, and how do both scale with the number of concurrent proposers?}
\end{quote}

\subsection{Main contributions}
To answer this question, we perform a game-theoretic analysis: due to the possibility of transaction duplication, a proposer's utility from their block depends on the other proposers. Our threat model is deliberately simple and stated up front: the adversary announces a public, unconditional bribe paid to any proposer that excludes the target transaction; proposers are rational, do not collude (collusion is typically equivalent to a single proposer), and build their blocks simultaneously from a shared mempool; and user tips are exogenous. Richer adversaries (e.g., bribes conditional on successful censorship) and endogenous user bidding are natural extensions that we discuss but do not model.

We solve this game algorithmically for its unique symmetric mixed Nash equilibrium, which specifies the probability of including each transaction as a function of the bids. Our algorithm is a one-dimensional bisection, is efficient, and applies to a wide family of TFMs. A ``cut-and-paste'' argument then yields closed-form expressions for the bribe required to censor any target transaction with any target probability.

With this machinery, we compute the expected throughput, user payments, and cost of bribery for any bid distribution. On synthetic (Pareto) bids, we vary the number of proposers, the block congestion, the adversary's target censorship probability, and the amount of burn, comparing three TFMs: duplication-penalizing, even-split, and collective. The duplication-penalizing TFM consistently dominates the other two on both eCR and throughput. We then validate these findings on the empirical tip distribution of $10{,}000$ recent Ethereum blocks.

Quantitatively, we find that eCR scales linearly with the number of proposers, while throughput quickly plateaus near the optimistic bound achieved by uncoordinated uniform sampling; thus, beyond the one-time duplication cost of moving away from a single proposer, adding proposers buys censorship resistance at essentially no additional throughput cost. This scaling survives fee burning: at Ethereum's early-2026 burn levels (base fee ${\approx}\,8.7\times$ the average tip), a single proposer yields $\text{eCR} \approx 0.09$, whereas $49$ concurrent proposers yield $\text{eCR} \approx 2.5$; the adversary must outspend the user rather than pay a few percent of their fee. Overall, our results show that MCP, equipped with the right TFM, is a promising structural remedy to economic censorship.

\section{Related Work}
\label{sec:related-work}
\paragraph{Protocol-level censorship resistance.}
Classical analyses of blockchain protocols establish \emph{liveness}: every valid transaction is eventually included~\cite{GarayKiayiasLeonardos2015,yin-hotstuff-2019,miller-honeybadger-2016}; see~\cite{AlposDavidKamarinakisZindros2025} for a survey of the censorship-resistance guarantees of widely-used consensus protocols. Eventual inclusion, however, is too weak for financial applications, motivating \emph{short-term} censorship resistance in the Byzantine model. BigDipper~\cite{xue-bigdipper-2023} adds short-term censorship resistance to leader-based BFT, and Abraham, Efron, and Ren~\cite{AbrahamEfronRen2025Latency} prove inherent latency costs of censorship resistance. Closest to our setting, concurrent work by Elsheimy et al.~\cite{ElsheimyKWPZ2026} formalizes the censorship-resistance/throughput trade-off in multi-proposer BFT protocols with up to $f$ Byzantine parties, and designs transaction-assignment rules that navigate this trade-off. This line of work bounds the \emph{worst-case delay} that faulty parties can impose; in contrast, we assume all proposers are rational and instead \emph{price} censorship under bribery. The two perspectives are complementary: Byzantine analyses cover arbitrary deviations by a bounded fraction of parties, while ours covers economically-motivated deviations by all parties.

\paragraph{Transaction fee mechanisms.}
Transaction fee mechanisms (TFMs) in blockchain systems have evolved from first-price auctions (as in Bitcoin) to hybrid models like Ethereum’s EIP-1559. Roughgarden’s~\cite{Roughgarden2021TFM} analysis of EIP-1559 showed that it steers the system toward market-clearing prices under demand, but also that no mechanism can be both fully strategy-proof and resistant to collusion or bribery. Chung and Shi~\cite{ChungShi2023TFM} sharpened this into a formal impossibility: no non-trivial TFM can simultaneously be incentive-compatible for users and for the block producer while resisting producer--user collusion. These results define a benchmark for later designs: TFMs must trade off between efficiency, incentive-compatibility, and resistance to strategic censorship or bribery.

Recent research has begun to extend TFM design to multi-proposer settings. One line of work, such as FOCIL (Fork-Choice Enforced Inclusion Lists), uses committees of validators to construct transaction inclusion lists before block building~\cite{noauthor_fork-choice_2024}. However, how to fairly allocate fees among multiple proposers remained an open question. Stouka et al.~\cite{StoukaMaThiery2025MultipleProposerTFM} introduce a game-theoretic model for multi-proposer blockchains (motivated by Ethereum’s FOCIL design). They extend Roughgarden’s canonical model by adding stages where several proposers each choose which transactions to include and by introducing a bribing adversary that can pay proposers to censor a transaction. Using this framework, they design and evaluate concrete multi-proposer TFMs (e.g. “Double” and “Single” TFMs) that maintain the Bayesian incentive-compatibility guarantees of EIP-1559 while greatly raising the cost of censorship. In effect, Stouka et al. show that appropriate fee-splitting rules can preserve honest behavior and enforce inclusion even under bribery. Notably, their adversary can condition bribes on successful censorship, which is generally cheaper for the adversary than the unconditional bribes we analyze; extending our equilibrium computation to conditional bribes is a natural direction for future work (see~\cref{sec:conclusion}).

 Fox et al.~\cite{Fox2023Auctions} and others quantify the “cost of censorship” in economic terms and demonstrate that multi-proposer TFMs can raise this cost by distributing responsibility for inclusion. Some mechanisms reward proposers based on the number of independent inclusions of a transaction, effectively discouraging selective omission~\cite{Fox2023Auctions,StoukaMaThiery2025MultipleProposerTFM}. At the application layer, Alpos et al.~\cite{AlposHeimbachNayakWadhwa2025SealedBid} design censorship-resistant sealed-bid auctions, complementing the auction-driven motivation of~\cite{Fox2023Auctions}.

\paragraph{Bribery attacks and censorship in practice.}
Bribery attacks on consensus have a long history: Bonneau~\cite{Bonneau2016Bribery} observed that adversaries can \emph{rent} rather than buy consensus power, McCorry et al.~\cite{McCorryHicksMeiklejohn2018} constructed smart contracts that trustlessly bribe miners, Winzer et al.~\cite{WinzerHerdFaust2019} analyzed temporary censorship attacks against rational miners, Sun et al.~\cite{SunRuanSu2020} give a practical quantification of bribery attacks, and Karakostas et al.~\cite{KarakostasKiayiasZacharias2024} study the efficacy of counterincentives. Closest to our notion of economic censorship, Berger et al.~\cite{berger} study bribery attacks against optimistic-rollup fraud proofs. Their setting considers repeated censorship over many sequential blocks, whereas we study the cost of censoring a transaction within a single block under rational proposer behavior. The two works are complementary: their analysis assumes a per-block censorship cost, while our eCR metric quantifies exactly that cost. Nor is censorship hypothetical: beyond the measurements of Wahrst\"atter et al.~\cite{WahrstaetterErnstbergerYaish2024BlockchainCensorship}, empirical studies of Ethereum's proposer--builder separation document widespread filtering by builders and relays~\cite{HeimbachKifferTorresWattenhofer2023PBS}, heavy concentration of the builder market~\cite{YangNayakZhang2025Builder}, and the centralizing role of private order flow~\cite{WangEtAl2024PrivateOrderFlow}---the latter directly motivating the private-mempool extension we discuss in~\cref{sec:conclusion}.
 
In parallel, new multi-leader consensus protocols have been proposed to boost throughput and resilience. Building on DAG mempools and consensus~\cite{DanezisNarwhal2022,spiegelman_bullshark_2022,GaiNiuEtAl2023}, DAG-based or “leaderless” designs (such as those used in Avalanche~\cite{avalanche2020docs}, Sui’s Mysticeti~\cite{babel_mysticeti_2024}, or Aptos~\cite{aptos2025whitepaper}) allow every validator to propose blocks simultaneously, vastly improving throughput compared to single-leader chains. Garimidi, Neu, and Resnick~\cite{mcp} make the general case for multiple concurrent proposers and design a protocol achieving selective-censorship resistance together with transaction hiding. Garimidi et al.~\cite{garimidi2025transactionfeemechanismdesign} formalize fee mechanisms for such leaderless protocols: they define a “strongly BPIC” equilibrium concept and propose the FPA-EQ mechanism, which splits first-price auction revenue equally among all contributing blocks. They prove FPA-EQ is strongly incentive-compatible and guarantees at least a 63.2\% fraction of the optimal welfare at equilibrium. These works (and others on Avalanche and Narwhal) highlight that multiple concurrent proposals require fundamentally different TFM designs, but can achieve higher performance and improved censorship resistance if fees are shared appropriately.

Collectively, these works highlight a growing interest in mechanisms that jointly optimize censorship resistance, incentive compatibility, and throughput in decentralized settings. Our work contributes to this space by modeling bribery and duplication in MCP systems and analyzing the equilibrium effects of different TFMs.

\paragraph{Comparison to inclusion lists.}
Inclusion lists are a powerful tool for censorship resistance. They originate from a line of Ethereum research on censorship under proposer--builder separation~\cite{Buterin2021CRunderPBS}, refined through forward inclusion lists~\cite{DAmato2022ForwardIL}, EIP-7547~\cite{EIP7547}, the ``no free lunch'' design~\cite{ButerinNeuder2023NoFreeLunch}, and committee-enforced inclusion sets~\cite{ThieryDAmatoMonnot2024COMIS}, culminating in FOCIL~\cite{noauthor_fork-choice_2024}, now drafted as EIP-7805~\cite{EIP7805}. Their main strength lies in their strictness; proposers sacrifice all of their reward for censoring transactions on the inclusion list. This is also a weakness. Proposers in certain jurisdictions may be legally required to censor certain transactions. By preventing this, inclusion lists may result in more geographic centralization.

A related recent proposal is AUCIL (Auction-based Inclusion Lists)~\cite{wadhwa_aucil_2025}, which is designed for committee-based inclusion. In AUCIL, a committee of proposers each compiles an input list of transactions where the number of overlaps is based on the fee paid by the transaction, and then an on-chain auction algorithm merges these lists into a final committed inclusion set. The key idea is that by having each committee member assigned certain transactions (using a correlated equilibrium based algorithm), the protocol ensures that high-paying transactions appear on at least one list. An adversary wishing to censor a transaction must therefore bribe all proposers who listed it, which is prohibitively expensive. In practice, AUCIL's auction step selects as many of the committee input lists as possible to form the block, thereby ``authentically'' binding each proposer to its list and making censorship very costly. This mechanism complements the multi-proposer TFMs above by enforcing that committed transactions are actually included, further strengthening censorship resistance in multi-proposer or proposer-builder systems. As a point of comparison to our design, AUCIL employs the even-split TFM, one of the admissible TFMs analyzed in this work.

In contrast, MCP presents a weaker, but more flexible, tool for censorship resistance. Proposers can still censor transactions if legally necessary, but that makes such transactions even more lucrative for proposers who can include it, helping to avoid geographic centralization. Also, MCP is desired for other reasons; namely, it can reduce MEV by limiting the ability of proposers to control the exact ordering of the final block. In that sense, our work analyzes the ``free'' censorship resistance that comes along with MCP; beyond the consensus protocol and TFM, no additional mechanism is required.

\section{Model and Definitions}
\label{sec:model}
We now explain our model. There are $n$ proposers who collectively build a block of size~$B$. Each proposer submits a sub-block of at most $\ell := B/n$ transactions (we assume for simplicity that $n$ divides $B$). These transactions are chosen from a \textit{shared mempool}. We denote this by the sorted list $\{t_1, t_2, \dots, t_m\}$, where $t_i$ refers to the tip of transaction $i$, and $t_1 \ge t_2 \ge \dots \ge t_m$. We sometimes abuse notation to use $t_i$ to refer to the tip or the transaction itself, when it is clear from the context.

The union of the proposers' sub-blocks form the final list of executed transactions. The TFM then portions out the proposers' payments as a function of the full list of sub-blocks. We consider the following three TFMs in this paper:
\begin{enumerate}
    \item the \textit{duplication-penalizing} TFM. If a transaction $t_i$ is included by only one proposer, then the proposer who included it is paid the full $t_i$. If a transaction is included by more than one proposer, then the transaction is not charged, and no tip is given to the proposers.
    \item the \textit{even-split} TFM. If a transaction $t_i$ is included by $k$ different proposers, then it is charged $t_i$, and every proposer that includes $t_i$ is paid $t_i/k$.
    \item the \textit{collective} TFM. All transactions are charged their tips, which are pooled and shared equally among \textit{all} proposers.
\end{enumerate}
Further, we assume that the TFM can specify a \textit{base fee} $\beta$ that is burnt. So, each included transaction pays a flat fee $\beta$ as well as its payment from the auction.

\begin{remark}[User payments and revenue]\label{rem:user-side}
Our analysis takes tips as exogenous and studies proposer incentives; we do not model users' bidding behavior or user--proposer collusion. Two features of the duplication-penalizing TFM deserve emphasis in this light. First, duplicated transactions are executed free of charge, so this TFM forgoes fee revenue relative to the even-split TFM, which charges every included transaction; our eCR metric reflects this through lower expected user payments. Second, a user able to coordinate with even a single proposer could deliberately trigger duplication to avoid paying. The two-tier tips of Fox et al.~\cite{Fox2023Auctions}, of which our duplication-penalizing TFM is the extreme case, mitigate both effects by charging a small tip upon duplicated inclusion; such intermediate TFMs are admissible in the sense of \cref{def:admissible}, so our equilibrium framework applies to them verbatim. We analyze the extreme case for clarity, and leave a user-side equilibrium analysis to future work.
\end{remark}

We assume the presence of an adversary who tries to censor a single transaction, $t_c$. The adversary offers an unconditional bribe of $b$ to all proposers who do not include $t_c$. This bribe is common knowledge; all proposers know that all other proposers have been offered the same bribe. 

The utility of a proposer is the sum of their payments, both from transactions and possibly from the bribe. Each proposer's strategy space is all subsets of size $B/n$ from the mempool. Our solution concept is a symmetric, mixed Nash equilibrium, which specifies the probability of sampling \textit{each subset} of transactions. But for our purposes, the only relevant feature of the equilibrium is its \textit{induced probabilities} over the transactions, rather than the full distribution over all subsets. We thus represent the equilibrium as $\mathbf{p} = (p_1, \dots, p_m)$, where $p_i$ is the probability that the proposer includes transaction $i$ in equilibrium.

For the following definitions, we fix the number of proposers, $n$, the total block size, $B$, the mempool, $\{\tx_1, \dots, \tx_m\}$, and the TFM, $\mathbf{T}$, in which each proposer that includes a transaction $\tx_i$ receives $\mathbf{T}(\tip_i, n_i)$, where $\tip_i$ is the transaction's tip and $n_i$ is the number of \emph{other} proposers that also include it (so a sole includer receives $\mathbf{T}(\tip_i, 0)$). In certain cases, even without including the transaction, the proposer can gain a certain value.  This induces a unique equilibrium profile, $\mathbf{p}$, as we will show later.
\begin{definition}[fixed-mempool throughput]\label{def:throughput-fixed}
    The \textit{throughput} of the current mempool is the expected number of unique transactions in the final block (given that all proposers play $\mathbf{p}$), divided by $\min(B, m)$.
\end{definition}
So, the maximum throughput occurs when all transactions are included, or, if that is not possible, when the block is full with no duplicate transactions. 

Finally, we turn to economic censorship resistance, defining it first on a per-transaction basis, assuming a fixed mempool state.
\begin{definition}[transaction cost]\label{def:tx-cost}
    The \textit{transaction cost} of $\tx_i$ is its expected payment \emph{conditional on inclusion} in the final block: the flat base fee $\beta$ plus the expected auction payment of $\tx_i$, in the equilibrium induced by the mempool state and the TFM, conditioned on the event that $\tx_i$ appears in the final block.
\end{definition}
This is the quantity that a user who ends up included actually pays; transactions that are excluded pay nothing, and normalizing by the conditional payment makes the eCR of transactions with different inclusion probabilities comparable.
To define the cost of bribery, we first assume that the adversary's goal is for $t_c$ to appear in the final block with probability at most $1-s$. Thus, $s$ is the adversary's desired rejection probability. For ease of exposition, we typically assume that $s = 0.9$, but we later show how our simulation results vary with $s$. We also rely on the fact that, for any bribe $b$, there is a unique equilibrium profile $\mathbf{p}^b$, which we show later.
\begin{definition}[bribery cost]\label{def:bribery-cost}
    Suppose an adversary wishes for $t_c$ to be censored with probability $s$, and let $b^*$ be the minimum per-proposer bribe such that, under the induced equilibrium profile $\mathbf{p}^{b^*}$, the probability that $t_c$ is included in the final block is at most $1-s$. The \textit{bribery cost} of transaction $c$ is the adversary's \emph{expected total expenditure}, $n\,(1 - p_c^{b^*})\,b^*$: the bribe is offered to each of the $n$ proposers, and each collects it (by excluding $t_c$) with probability $1 - p_c^{b^*}$.
\end{definition}
\begin{definition}[per-tx eCR]~\label{def:ecnomic-cr-per-tx}
    The economic censorship resistance of transaction $t_i$ is its bribery cost divided by its transaction cost.
\end{definition}
The intuition behind our definition is that transactions which cost more are likely from users with higher valuations, and should thus be harder to censor. Put differently, a bribery cost of 10 and a transaction cost of 1 represents strong censorship resistance; the adversary must spend 10 times the user. If the transaction cost is itself 10, this is no longer as resistant to censorship.

Since all of these definitions depend on the state of the mempool, we add a distributional assumption. So suppose all $t_i$'s are drawn iid from distribution $F$. 
\begin{definition}[throughput]
    The expected throughput under $F$ is the expectation, over mempools $\mathbf{t} \sim F^m$, of the throughput of $\mathbf{t}$.
\end{definition}
\begin{definition}[eCR]
    The expected censorship resistance under $F$ is the expectation, over mempools $\mathbf{t} \sim F^m$, of the average economic censorship resistance of a transaction $t_i \in \mathbf{t}$.
\end{definition}

\section{A Generalized Equilibrium Algorithm}
\label{sec:algorithm}
In this section, we do not consider a fixed TFM or bribe. Instead, we provide an algorithm that solves the symmetric mixed Nash equilibrium, given \textit{payoff functions}, that describe how the utility of a proposer changes as they include a transaction, given the other proposers' probability of including that transaction. Our algorithm works for a range of admissible payoff functions, that include all three TFMs we consider, in the absence or presence of any bribe. Intuitively, an admissible payoff function is one where the utility of including transaction $i$ decreases as the probability that others include that transaction increases. In the simplest case of the duplication-penalizing TFM, this holds as the utility of including a transaction is only nonzero if no one else includes it. 

\subsection{Symmetric Equilibrium via a Linear Program}
\label{sec:lp-characterization}

We characterize symmetric mixed-strategy equilibria of the proposer game as optima of a linear program, and use linear-programming duality to derive a relation between the fee paid by transactions and the probability of its selection. This relation, that all transactions assigned interior probability share a common marginal payoff $\lambda^*$ is then used to prove a symmetric ($\epsilon-$approximate) equilibrium in \Cref{alg:binsearch}.

\paragraph{Setup.}
Consider a set of transactions $\tx_1, \dots, \tx_\numtransactions$, each with a single tip value $\tip_i$. We seek a symmetric mixed strategy: a probability vector $\vec p = (p_1, \dots, p_\numtransactions) \in [0,1]^\numtransactions$ used by every proposer, where $p_i$ is the probability of including $\tx_i$. If, in addition to the proposer under consideration, $n_i$ \emph{other} proposers include $\tx_i$, the proposer's payment for it is $\mathbf{T}(t_i, n_i)$, following the convention of \cref{sec:model}. We assume throughout that $m > \ell$; otherwise the block trivially fits the entire mempool.

For each transaction $\tx_i$, let $\overline{\lambda}_i$ denote the maximum possible payoff to a proposer (usually when it is the sole includer of $\tx_i$), and $\underline{\lambda}_i$ the minimum possible payoff (usually when all proposers include $\tx_i$). Fix a profile of beliefs $\vec p^* \in [0,1]^\numtransactions$ representing what every \emph{other} proposer will do. 

From a proposer's perspective, the expected marginal payoff of including $\tx_i$ is a function of the probability that other proposers select this transaction, and $\tip_i$. Let $u_i$ be the marginal utility for the proposer (utility for including the transaction minus utility for not including the transaction) when including the transaction $\tx_i$. If every \emph{other} proposer includes $\tx_i$ independently with probability $p_i$, then
\begin{align}
\label{eq:marginal-payoff}
u_i = \sum_{n_i = 0}^{n-1}\binom{n-1}{n_i} p_i^{n_i}(1-p_i)^{n-1-n_i}\mathbf{T}(t_i,n_i)
\end{align}
where, with slight abuse of notation, $\mathbf{T}(t_i, n_i)$ in~\eqref{eq:marginal-payoff} denotes the \emph{marginal} payoff---the difference in a proposer's payment between including and not including $\tx_i$---given that $n_i$ others include it. For the duplication-penalizing and even-split TFMs this coincides with the payment itself; for the collective TFM it is the change in the proposer's pool share (\Cref{ex:socialist}).

\paragraph{Best response as a linear program.} \emph{Treating $\vec p^*$ (and hence each $u_i$) as fixed}, a deviator's best response maximizes total expected payoff subject to a block-size budget $\ell$:
\begin{align}
  \label{eq:primal}
  \max_{\vec p} \quad & \sum_{i=1}^\numtransactions u_i\, p_i \\
  \text{s.t.}\quad
  & \sum_{i=1}^\numtransactions p_i = \ell, \notag \\
  & 0 \le p_i \le 1 && \forall\, i \in \{1, \dots, \numtransactions\}. \notag
\end{align}
A symmetric equilibrium then corresponds to a fixed point at which $\vec p = \vec p^*$ is itself optimal for~\eqref{eq:primal} with the marginal payoffs $u_i$ induced by $\vec p^*$ via~\eqref{eq:marginal-payoff}. (We impose the block-size budget with equality; since payments are non-negative, only a bribed transaction can have negative marginal utility, and as $m > \ell$ there are always at least $\ell$ transactions with non-negative marginal utility to fill the sub-block, so this is without loss of generality.)

\begin{property}[Equilibrium structure]\label{prop:equilibrium}
There exists an optimal solution $\vec p^*$ to~\eqref{eq:primal} and a scalar $\lambda^* \in \R$ such that, for every $i \in \{1, \dots, \numtransactions\}$:
\begin{enumerate}
  \item if $u_i > \lambda^*$, then $p_i^* = 1$;
  \item if $u_i < \lambda^*$, then $p_i^* = 0$;
  \item if $0 < p_i^* < 1$, then $u_i = \lambda^*$.
\end{enumerate}
Conversely, any feasible $\vec p$ for which some $\lambda$ satisfies conditions 1--3 is optimal for~\eqref{eq:primal}.
\end{property}

\begin{proof}
We establish the claim through linear-programming duality. Introduce dual variables $a_i \ge 0$ for the upper-bound constraints $p_i \le 1$ and $\lambda \in \R$ for the equality constraint $\sum_i p_i = \ell$. The dual program is
\begin{align}
  \label{eq:dual}
  \min_{\vec a,\, \lambda} \quad & \sum_{i=1}^\numtransactions a_i + \ell\, \lambda \\
  \text{s.t.}\quad
    & a_j + \lambda \ge u_j && \forall\, j \in \{1, \dots, \numtransactions\}, \notag \\
    & a_i \ge 0 && \forall\, i \in \{1, \dots, \numtransactions\}, \notag \\
    & \lambda \in \R. \notag
\end{align}
The primal~\eqref{eq:primal} is feasible (e.g., $p_i = \ell / \numtransactions$ when $\ell \le \numtransactions$) and bounded, so by strong duality both programs admit optima with equal value. Let $(\vec p^*, \vec a^*, \lambda^*)$ denote a pair of optimal primal and dual solutions. Complementary slackness gives, for each $i$,
\begin{align}
  \label{eq:cs-1}
  a_i^*\, (1 - p_i^*) &= 0, \\
  \label{eq:cs-2}
  (a_i^* + \lambda^* - u_i)\, p_i^* &= 0.
\end{align}
We now verify each of the three claims.

\noindent\textbf{Case 1: $u_i > \lambda^*$.}
Suppose for contradiction that $p_i^* < 1$. Then~\eqref{eq:cs-1} forces $a_i^* = 0$, but dual feasibility requires $a_i^* + \lambda^* \ge u_i$, i.e., $a_i^* \ge u_i - \lambda^* > 0$, a contradiction. Hence $p_i^* = 1$.

\noindent\textbf{Case 2: $u_i < \lambda^*$.}
Suppose for contradiction that $p_i^* > 0$. Then~\eqref{eq:cs-2} forces $a_i^* + \lambda^* - u_i = 0$, i.e., $a_i^* = u_i - \lambda^* < 0$, contradicting dual feasibility $a_i^* \ge 0$. Hence $p_i^* = 0$.

\noindent\textbf{Case 3: $0 < p_i^* < 1$.}
Since $p_i^* < 1$,~\eqref{eq:cs-1} gives $a_i^* = 0$. Substituting into~\eqref{eq:cs-2} and using $p_i^* > 0$ yields $\lambda^* - u_i = 0$, i.e., $u_i = \lambda^*$.

For the converse, suppose feasible $\vec p$ and some $\lambda$ satisfy conditions 1--3, and set $a_i := \max(u_i - \lambda, 0) \ge 0$. The pair $(\vec a, \lambda)$ is dual feasible, and conditions 1--3 imply the complementary slackness equations \eqref{eq:cs-1}--\eqref{eq:cs-2}: if $p_i < 1$ then $u_i \le \lambda$ (conditions 2--3), so $a_i = 0$; and if $p_i > 0$ then $u_i \ge \lambda$ (conditions 1 and 3), so $a_i + \lambda = u_i$. Hence $\vec p$ is optimal. \qedhere

\end{proof}

\Cref{prop:equilibrium} reduces the problem of finding a symmetric equilibrium to identifying the threshold $\lambda^*$ and the corresponding inclusion probabilities. The next subsection presents an efficient binary-search algorithm for this task.

\subsection{A Binary-Search Algorithm for Symmetric Equilibria}
\label{sec:binsearch}

Building on \Cref{prop:equilibrium}, we now give an efficient algorithm to compute symmetric equilibria. The idea is to invert the marginal-payoff relationship~\eqref{eq:marginal-payoff}: at a symmetric fixed point $\vec p = \vec p^*$, equation~\eqref{eq:marginal-payoff} couples each $p_i$ to $u_i$, and \Cref{prop:equilibrium} couples each interior $p_i$ to a common threshold $\lambda^*$. Composing these, the equilibrium probability of each transaction is determined by $\lambda^*$ alone, reducing the search to a one-dimensional bisection over $\lambda$.

\paragraph{Setup.} For each transaction $\tx_i$, let $f_i(\lambda)$ denote the symmetric inclusion probability that makes a proposer indifferent at marginal payoff $\lambda$: the value of $p_i$ for which~\eqref{eq:marginal-payoff} yields $u_i = \lambda$ when every other proposer also plays $p_i$. Concretely, for the marginal-payoff form of \Cref{ex:socialist}, $f_i$ admits a closed form; in general it can be computed by inverting~\eqref{eq:marginal-payoff} numerically. By construction, $f_i$ takes the value $1$ at $\lambda = \underline{\lambda}_i$ (full competition) and $0$ at $\lambda = \overline{\lambda}_i$ (monopoly inclusion), matching the bounds defined in~\Cref{sec:lp-characterization}.

\begin{definition}[Admissible payoff function]\label{def:admissible}
A function $f_i \colon \R \to \R$ is \emph{admissible} with parameters $(\underline{\lambda}_i, \overline{\lambda}_i)$ if:
\begin{itemize}
    \item \textbf{(A1)} \textbf{Boundary values:} $f_i(\underline{\lambda}_i) = 1$ and $f_i(\overline{\lambda}_i) = 0$.
    \item \textbf{(A2)} \textbf{Continuity and strict monotonicity:} $f_i$ is continuous and strictly decreasing on $[\underline{\lambda}_i, \overline{\lambda}_i]$.
    \item\textbf{(A3)} \textbf{Range:} $f_i(\lambda) > 1$ for $\lambda < \underline{\lambda}_i$ and $f_i(\lambda) < 0$ for $\lambda > \overline{\lambda}_i$.
\end{itemize}
\end{definition}

\begin{example}[Collective TFM]\label{ex:socialist}
Suppose all fees from included transactions are pooled and divided equally among the $n$ proposers, regardless of who included them. Proposer $j$'s total payoff is
\[
U_j \;=\; \frac{1}{n} \sum_{i=1}^{\numtransactions} T_i \cdot \mathbf{1}\bigl[\tx_i \text{ included by at least one proposer}\bigr].
\]
Because $j$ receives the same share of $T_i$ whether or not she personally includes $\tx_i$ (provided someone does), the \emph{marginal} value of inclusion equals the expected pool contribution conditional on every other proposer omitting $\tx_i$:
\[
\Delta u_i(p_i) \;=\; \frac{T_i}{n}\,(1 - p_i)^{n-1}.
\]
Setting $\Delta u_i(p_i) = \lambda$ and inverting yields
\[
f_i(\lambda) \;=\; 1 - \left(\frac{n\lambda}{T_i}\right)^{\!\!\frac{1}{n-1}},
\qquad \underline{\lambda}_i = 0, \qquad \overline{\lambda}_i = \frac{T_i}{n}.
\]
\emph{Verification of admissibility.}
\textbf{(A1)} $f_i(0) = 1$ and $f_i(T_i/n) = 0$.
\textbf{(A2)} The map $\lambda \mapsto (n\lambda/T_i)^{1/(n-1)}$ is strictly increasing for $\lambda \ge 0$, so $f_i$ is strictly decreasing on $[0, T_i/n]$.
\textbf{(A3)} For $\lambda > T_i/n$ the inner term exceeds $1$, so $f_i(\lambda) < 0$; the natural extension to $\lambda < 0$ gives $f_i(\lambda) > 1$.
\end{example}

\begin{example}[Duplication-penalizing TFM]\label{ex:duplication}
Under this TFM a transaction pays its tip $T_i$ to its includer only when it is included exactly once, so the marginal value of including $\tx_i$ is nonzero only in the event that no other proposer includes it:
\[
u_i(p_i) \;=\; (1 - p_i)^{n-1}\, T_i,
\qquad
f_i(\lambda) \;=\; 1 - \left(\frac{\lambda}{T_i}\right)^{\!\!\frac{1}{n-1}},
\qquad \underline{\lambda}_i = 0, \qquad \overline{\lambda}_i = T_i.
\]
This is the same functional form as \Cref{ex:socialist} with $T_i$ in place of $T_i/n$, and admissibility is verified identically.
\end{example}

\begin{example}[Even-split TFM]\label{ex:even-split}
Suppose a transaction included by $k$ proposers is charged its full tip $T_i$, which is split evenly, so that each includer receives $T_i/k$. If every other proposer includes $\tx_i$ independently with probability $p_i$, the number of \emph{other} includers is $K \sim \mathrm{Bin}(n-1, p_i)$, and the marginal value of inclusion is
\[
u_i(p_i) \;=\; T_i \cdot \E\!\left[\frac{1}{1+K}\right]
\;=\; T_i \sum_{k=0}^{n-1}\binom{n-1}{k}\frac{p_i^{k}(1-p_i)^{n-1-k}}{k+1}
\;=\; T_i\,\frac{1-(1-p_i)^{n}}{n\,p_i},
\]
where the last equality uses $\binom{n-1}{k}\frac{1}{k+1} = \frac{1}{n}\binom{n}{k+1}$ and the binomial theorem, and $u_i(0) := T_i$ by continuity.

\emph{Verification of admissibility.}
\textbf{(A1)} $u_i(0) = T_i$ and $u_i(1) = T_i/n$, so the inverse $f_i$ of $u_i$ satisfies $f_i(\overline{\lambda}_i) = 0$ and $f_i(\underline{\lambda}_i) = 1$ with $\overline{\lambda}_i = T_i$ and $\underline{\lambda}_i = T_i/n$.
\textbf{(A2)} Increasing $p_i$ makes $K$ stochastically larger, and $k \mapsto 1/(1+k)$ is strictly decreasing, so $u_i$ is strictly decreasing on $[0,1]$ for every $n \ge 2$; hence its inverse $f_i$ is strictly decreasing on $[\underline{\lambda}_i, \overline{\lambda}_i]$.
\textbf{(A3)} Extend $f_i$ outside $[\underline{\lambda}_i, \overline{\lambda}_i]$ by any strictly decreasing continuous extension (e.g., affinely); the algorithm only ever uses the clamped function $\fclamp_i$, which is unaffected by the choice of extension.

Unlike the collective TFM, $u_i$ here admits no closed-form inverse. This poses no obstacle: \Cref{alg:binsearch} only requires \emph{pointwise evaluation} of $f_i$, and since $u_i$ is continuous and strictly decreasing, $f_i(\lambda)$ can be evaluated to any accuracy $\varepsilon'$ by an inner bisection on $[0,1]$, adding a factor $O(\log(1/\varepsilon'))$ to the per-evaluation cost in \Cref{thm:termination}.

\emph{A tractable approximation.} Replacing $\E[1/(1+K)]$ by $1/(1+\E[K])$ yields
\[
\tilde{u}_i(p_i) \;=\; \frac{T_i}{1 + (n-1)p_i},
\qquad
\tilde{f}_i(\lambda) \;=\; \frac{\frac{T_i}{\lambda} - 1}{n-1},
\]
which shares the boundary values $\tilde{u}_i(0) = T_i$ and $\tilde{u}_i(1) = T_i/n$ and has a closed-form inverse. By Jensen's inequality, $\tilde{u}_i(p_i) \le u_i(p_i)$, so this approximation \emph{understates} the value of inclusion at interior points. Since $\tilde{f}_i$ is itself admissible with the same parameters, all of our structural results apply to it verbatim; we use it only for intuition and closed-form expressions, while all reported numerical results are computed with the exact $u_i$.
\end{example}

\paragraph{The clamped aggregate.} For any candidate $\lambda \in \R$, define the \emph{clamped total inclusion mass}
\[
G(\lambda) = \sum_{i=1}^{\numtransactions} \min\bigl(\max(f_i(\lambda),\, 0),\, 1\bigr).
\]
By admissibility, the $i$th term equals $1$ for $\lambda \le \underline{\lambda}_i$, equals $0$ for $\lambda \ge \overline{\lambda}_i$, and is strictly decreasing in between. Hence $G$ is non-increasing on $\R$, with $G(\min_i \underline{\lambda}_i) = \numtransactions$ and $G(\max_i \overline{\lambda}_i) = 0$. Combining \Cref{prop:equilibrium} with the definition of $f_i$, finding a symmetric equilibrium reduces to locating $\lambda^*$ with $G(\lambda^*) = \ell$, where $\ell$ is the target block size. Since $G$ is monotone, a one-dimensional bisection suffices.

\begin{remark}[Marginals suffice]\label{rem:marginals}
Two observations justify working with inclusion probabilities rather than full distributions over sub-blocks. First, since proposers randomize independently and a proposer's payoff is additive over transactions, with the term for $\tx_i$ depending on the others only through the number of them that include $\tx_i$ (distributed $\mathrm{Bin}(n-1, p_i)$), opponents' strategies affect a proposer's payoff only through their marginals; hence the best-response problem over distributions on $\ell$-subsets reduces to the linear program~\eqref{eq:primal} over marginals. Second, every feasible point of~\eqref{eq:primal} is realizable: for any $\vec p \in [0,1]^\numtransactions$ with $\sum_i p_i = \ell \in \mathbb{N}$ there is a random $\ell$-subset $S$ with $\Pr[i \in S] = p_i$. (Systematic sampling: lay consecutive intervals of lengths $p_1, \dots, p_\numtransactions$ on a circle of circumference $\ell$, draw $x$ uniformly from $[0,1)$, and select the transactions whose intervals contain a point $x + k$ for $k = 0, \dots, \ell-1$; since each $p_i \le 1$, no interval contains two such points, and exactly $\ell$ transactions are selected.) Together, these imply that optimal marginal vectors correspond exactly to best responses in the original strategy space.
\end{remark}

\begin{theorem}[Existence and uniqueness]\label{thm:uniqueness}
For any set of transactions $\tx_1, \dots, \tx_\numtransactions$ with admissible payoff functions $f_1, \dots, f_\numtransactions$ and any target block size $\ell \in (0, \numtransactions)$, there exists $\lambda^*$ with $G(\lambda^*) = \ell$, and $\mathbf{p}^* = (\fclamp_1(\lambda^*), \dots, \fclamp_\numtransactions(\lambda^*))$ is a symmetric mixed Nash equilibrium. Moreover, the equilibrium profile is unique: any two thresholds with clamped inclusion mass $\ell$ induce the same probability vector $\mathbf{p}^*$.
\end{theorem}

\begin{proof}
\emph{Existence.} Each $\fclamp_i$ is continuous on $\R$, by continuity of $f_i$ on $[\underline{\lambda}_i, \overline{\lambda}_i]$ \textbf{(A2)} and the matching boundary values \textbf{(A1)}; hence $G$ is continuous and non-increasing, with $G(\min_i \underline{\lambda}_i) = \numtransactions > \ell$ and $G(\max_i \overline{\lambda}_i) = 0 < \ell$. By the intermediate value theorem there exists $\lambda^*$ with $G(\lambda^*) = \ell$. When every proposer plays $\mathbf{p}^* = (\fclamp_i(\lambda^*))_i$, the induced marginal payoffs satisfy the three conditions of \Cref{prop:equilibrium} at threshold $\lambda^*$: if $\fclamp_i(\lambda^*) = 1$ then $\lambda^* \le \underline{\lambda}_i = u_i(1)$; if $\fclamp_i(\lambda^*) = 0$ then $\lambda^* \ge \overline{\lambda}_i = u_i(0)$; and at interior points $u_i(p_i^*) = \lambda^*$ by definition of $f_i$. By the converse direction of \Cref{prop:equilibrium}, $\mathbf{p}^*$ is a best response to itself, and by \Cref{rem:marginals} it is induced by a mixed strategy over $\ell$-subsets; hence a symmetric mixed Nash equilibrium exists.

\emph{Uniqueness.} By \Cref{prop:equilibrium}, any symmetric equilibrium profile has the form $(\fclamp_i(\lambda))_i$ for some threshold $\lambda$ with $\sum_i \fclamp_i(\lambda) = \ell$. Let $\lambda^* < \lambda'$ both satisfy this. Since each $\fclamp_i$ is non-increasing, $\fclamp_i(\lambda') \le \fclamp_i(\lambda^*)$ for every $i$; since both sums equal $\ell$, every inequality holds with equality, so the two thresholds induce the same probability vector. (The threshold itself need not be unique: $G$ may be constant on an interval on which every $f_i$ is clamped; the induced profile is unaffected.)

\end{proof}

\begin{algorithm}[htb!]
\caption{Symmetric-equilibrium binary search}
\label{alg:binsearch}
\begin{algorithmic}[1]
\Require Admissible functions $f_1, \dots, f_\numtransactions$ with parameters $(\underline{\lambda}_i, \overline{\lambda}_i)$; block size $\ell$; tolerance $\varepsilon > 0$.
\Ensure Probabilities $(p_1^*, \dots, p_\numtransactions^*)$ forming an $\varepsilon$-approximate symmetric equilibrium.

\Statex
\Function{$G$}{$\lambda$}\Comment{clamped total inclusion mass}
    \State $S \gets 0$
    \For{$i = 1, \dots, \numtransactions$}
        \If{$\lambda \le \underline{\lambda}_i$}
            \State $S \gets S + 1$ \Comment{$\tx_i$ always included}
        \ElsIf{$\lambda \ge \overline{\lambda}_i$}
            \State $S \gets S + 0$ \Comment{$\tx_i$ never included}
        \Else
            \State $S \gets S + f_i(\lambda)$ 
        \EndIf
    \EndFor
    \State \Return $S$
\EndFunction

\Statex
\Function{BinarySearch}{$f_1,\dots,f_\numtransactions,\ \ell,\ \varepsilon$}
    \If{$\numtransactions \le \ell$} \Return $(1, 1, \dots, 1)$ 
    \EndIf
    \If{$\ell = 0$} \Return $(0, 0, \dots, 0)$
    \EndIf
    \State $\lambda_{\min} \gets \min_i \underline{\lambda}_i$;\quad $\lambda_{\max} \gets \max_i \overline{\lambda}_i$
        \Comment{$G(\lambda_{\min}) = \numtransactions \ge \ell$, $G(\lambda_{\max}) = 0 \le \ell$}
    \Loop
        \State $\lambda_{\mathrm{mid}} \gets (\lambda_{\min} + \lambda_{\max})/2$
        \If{$|G(\lambda_{\mathrm{mid}}) - \ell| \le \varepsilon$}
            \State \textbf{break} \Comment{inclusion mass is within tolerance of $\ell$}
        \ElsIf{$G(\lambda_{\mathrm{mid}}) > \ell$}
            \State $\lambda_{\min} \gets \lambda_{\mathrm{mid}}$ \Comment{too many inclusions; raise $\lambda$}
        \Else
            \State $\lambda_{\max} \gets \lambda_{\mathrm{mid}}$ \Comment{too few inclusions; lower $\lambda$}
        \EndIf
    \EndLoop
    \State $\lambda^* \gets \lambda_{\mathrm{mid}}$
    \For{$i = 1, \dots, \numtransactions$}
        \State $p_i^* \gets \min\!\bigl(\max(f_i(\lambda^*),\, 0),\, 1\bigr)$
    \EndFor
    \State \Return $(p_1^*, \dots, p_\numtransactions^*)$
\EndFunction
\end{algorithmic}
\end{algorithm}

\begin{theorem}[Termination and runtime]\label{thm:termination}
Let $\omega_G$ be a modulus of continuity of $G$ on $[\min_i \underline{\lambda}_i,\, \max_i \overline{\lambda}_i]$ (one exists, since $G$ is continuous on this compact interval), and fix any $\delta_\varepsilon > 0$ with $\omega_G(\delta_\varepsilon) \le \varepsilon$. Then \Cref{alg:binsearch} terminates after at most $\bigl\lceil \log_2(\Delta/\delta_\varepsilon) \bigr\rceil$ iterations, where $\Delta = \max_i \overline{\lambda}_i - \min_i \underline{\lambda}_i$. Each iteration evaluates $G$ once at cost of $O(\numtransactions)$ evaluations of the $f_i$, giving total runtime $O\!\bigl(\numtransactions \log(\Delta/\delta_\varepsilon)\bigr)$.
\end{theorem}

\begin{proof}
By \Cref{thm:uniqueness} there exists $\lambda^*$ with $G(\lambda^*) = \ell$, and the update rule maintains the invariant $\lambda_{\min} \le \lambda^* \le \lambda_{\max}$ (verified in the proof of \Cref{thm:correctness}). The bracket width halves each iteration: after $k$ iterations it equals $\Delta/2^k$. Once $\Delta/2^k \le \delta_\varepsilon$, the midpoint satisfies $|\lambda_{\mathrm{mid}} - \lambda^*| \le \delta_\varepsilon$, hence
$|G(\lambda_{\mathrm{mid}}) - \ell| = |G(\lambda_{\mathrm{mid}}) - G(\lambda^*)| \le \omega_G(\delta_\varepsilon) \le \varepsilon$,
and the termination condition is met. Each iteration computes $G$ once, inspecting all $\numtransactions$ transactions; for TFMs whose $f_i$ admits no closed form (e.g., the even-split TFM of \Cref{ex:even-split}), each evaluation additionally costs an inner bisection, multiplying the runtime by $O(\log(1/\varepsilon'))$ for inner accuracy $\varepsilon'$.
\end{proof}

With existence and uniqueness established, we can now prove that \Cref{alg:binsearch} outputs an approximate equilibrium.

\begin{definition}[$\varepsilon$-approximate equilibrium]\label{def:approx-eq}
A probability vector $(p_1, \dots, p_\numtransactions) \in [0,1]^\numtransactions$ is an \emph{$\varepsilon$-approximate symmetric mixed Nash equilibrium} if there exists $\hat{\lambda} \in \R$ such that
\begin{enumerate}
    \item the equilibrium conditions of \Cref{prop:equilibrium} hold at $\hat{\lambda}$ with probabilities $p_i$, and
    \item $\bigl|\sum_{i=1}^{\numtransactions} p_i - \ell\bigr| \le \varepsilon$.
\end{enumerate}
\end{definition}

\begin{theorem}[Correctness]\label{thm:correctness}
The output $(p_1^*, \dots, p_\numtransactions^*)$ of \Cref{alg:binsearch} is an $\varepsilon$-approximate symmetric mixed Nash equilibrium, where $\varepsilon$ is the tolerance supplied to the algorithm.
\end{theorem}

\begin{proof}
By \Cref{thm:uniqueness}, there exists $\lambda^*$ with $G(\lambda^*) = \ell$; fix any such $\lambda^*$. The corresponding probability vector $\vec p^* = (\fclamp_1(\lambda^*), \dots, \fclamp_\numtransactions(\lambda^*))$ is the unique symmetric mixed Nash equilibrium profile. We prove correctness by maintaining the following invariant throughout the execution of the algorithm.

\paragraph{Invariant.} At every iteration of the loop, $\lambda_{\min} \le \lambda^* \le \lambda_{\max}$.

\paragraph{Base case.} Initially, $\lambda_{\min} = \min_i \underline{\lambda}_i$ and $\lambda_{\max} = \max_i \overline{\lambda}_i$. By admissibility, $G(\lambda_{\min}) = \numtransactions \ge \ell$ and $G(\lambda_{\max}) = 0 \le \ell$. Since $G$ is non-increasing and $G(\lambda^*) = \ell$, we have $\lambda_{\min} \le \lambda^* \le \lambda_{\max}$.

\paragraph{Inductive step.} Suppose $\lambda_{\min} \le \lambda^* \le \lambda_{\max}$ at the start of an iteration, and let $\lambda_{\mathrm{mid}} = (\lambda_{\min} + \lambda_{\max})/2$. If the termination condition $|G(\lambda_{\mathrm{mid}}) - \ell| \le \varepsilon$ is met, the loop exits. Otherwise, there are two cases:
\begin{itemize}
    \item If $G(\lambda_{\mathrm{mid}}) > \ell$: since $G$ is non-increasing and $G(\lambda^*) = \ell < G(\lambda_{\mathrm{mid}})$, we have $\lambda^* \ge \lambda_{\mathrm{mid}}$. The algorithm sets $\lambda_{\min} \gets \lambda_{\mathrm{mid}}$, so $\lambda_{\min} \le \lambda^* \le \lambda_{\max}$ is preserved.
    \item If $G(\lambda_{\mathrm{mid}}) < \ell$: since $G$ is non-increasing and $G(\lambda^*) = \ell > G(\lambda_{\mathrm{mid}})$, we have $\lambda^* \le \lambda_{\mathrm{mid}}$. The algorithm sets $\lambda_{\max} \gets \lambda_{\mathrm{mid}}$, so $\lambda_{\min} \le \lambda^* \le \lambda_{\max}$ is preserved.
\end{itemize}

\paragraph{Equilibrium conditions.} At termination, let $\hat\lambda = \lambda_{\mathrm{mid}}$ be the algorithm's output threshold, and set $p_i = \fclamp_i(\hat\lambda)$ for each $i$. We verify the three cases of \Cref{prop:equilibrium} using admissibility of each $f_i$.
\begin{itemize}
    \item \emph{Pure inclusion ($p_i = 1$).} This occurs iff $f_i(\hat\lambda) \ge 1$, equivalently $\hat\lambda \le \underline{\lambda}_i$ by~\textbf{(A1)} and~\textbf{(A3)}. In this regime the marginal payoff $u_i$ exceeds $\hat\lambda$ even under maximal competition, so always-include is a best response.
    \item \emph{Pure exclusion ($p_i = 0$).} This occurs iff $f_i(\hat\lambda) \le 0$, equivalently $\hat\lambda \ge \overline{\lambda}_i$. The marginal payoff is below $\hat\lambda$ even at monopoly, so always-exclude is a best response.
    \item \emph{Interior points ($p_i \in (0,1)$).} Here $\hat\lambda \in (\underline{\lambda}_i, \overline{\lambda}_i)$ and $p_i = f_i(\hat\lambda)$, so by definition of $f_i$, $u_i = \hat\lambda$. The proposer is exactly indifferent between any selection.
\end{itemize}

\paragraph{Block-size constraint.} The termination condition directly guarantees
\[
\Bigl|\sum_{i=1}^{\numtransactions} p_i - \ell\Bigr| = |G(\hat\lambda) - \ell| \le \varepsilon.
\]
Hence the output satisfies \Cref{def:approx-eq}.

\end{proof}

\section{Deriving Economic Censorship Resistance}
\label{sec:deriving-censorship-resistance}
In this section, we describe a general ``cut-and-paste'' method to derive the bribe required to censor a target transaction with a given probability. Our method only requires computing the equilibrium probabilities for a \emph{single} game. Using these probabilities, we can derive a closed-form expression for the bribe required, and by extension, for the economic censorship resistance (per~\cref{def:ecnomic-cr-per-tx}) of the target transaction. We then apply the cut-and-paste method to all three TFMs we study in this paper.

\subsection{Cut-and-Paste Method}

\paragraph{Setup}
 We recall some notation from earlier.
 Let $u_i(p_i)$ be the expected (increase in) utility of including $\tx_i$ assuming everyone else includes it with probability $p_i$. 
 Given transaction tips $\tvec = (t_i)_{i \in m}$ and target probability mass $\pmass$,  let $(\pvec = (p_i)_{i \in [m]}, \lambda) = \equil(\tvec, \pmass)$ be the equilibrium probabilities and common utility level, respectively.
This means that $u_i(p_i) = \lambda$ for all $i \in [m]$.
 Also let $f_i(x)$ be the admissible function (per~\cref{def:admissible}) such that $p_i = f_i(\lambda)$.

Suppose the adversary $\adv$ wishes to censor a target transaction $\tx_c$ (with tip $t_c$) for some $c \in [m]$.
We are interested in determining the bribe amount $b^*$ that $\adv$ needs to offer to all parties, so that $\tx_c$ is censored with probability $\censorprob \in (0,1]$. In order for $\adv$ to achieve the target censorship probability $\censorprob$, each party must include $\tx_c$ with probability $p_c^*$ such that $\censorprob = (1 - p_c^*)^n$, or equivalently, $p_c^* = 1 - \censorprob^{1/n}$. This holds regardless of the underlying TFM.

\paragraph{Cut-and-Paste Method}
The method proceeds in two steps. 

First, the \emph{cut} step: solve for the equilibrium of the game with 
$\tx_c$ removed, reserving exactly $p_c^*$ units of capacity to be filled 
in later. Let 
$(\pvec_{-c}, \lambda_{-c}) := \equil(\tvec_{-c},\, \pmass - p_c^*)$
denote this equilibrium.

Second, the \emph{paste} step: add $\tx_c$ back, with the adversary's bribe 
$b^*$ chosen so that, in the resulting equilibrium of the full game 
$\game(\tvec, \pmass)$, (1)~$\tx_c$ is included with the target probability 
$p_c^*$; and (2)~every $\tx_j \neq \tx_c$ retains the probability $p_j$ it 
had in $\pvec_{-c}$. We carry out this analysis in \cref{sec:cnp-correctness} and 
obtain the closed-form bribe $b^* = u_c(p_c^*) - \lambda_{-c}$.

\subsection{Correctness of the Cut-and-Paste Method}
\label{sec:cnp-correctness}

We prove that the cut-and-paste method produces an equilibrium of $\game(\tvec, \pmass)$ under bribe $b^*$ in which $\tx_c$ attains probability $p_c^*$ and every $\tx_j \neq \tx_c$ retains its probability from the post-cut equilibrium $\pvec_{-c}$. 

The argument relies on a structural property of equilibria captured by the following lemma: a new transaction can be added to the game without disturbing the existing inclusion probabilities or utility level, provided capacity is expanded by exactly the amount the new transaction consumes at level $\lambda$. The paste step will exploit this to add $\tx_c$ 
back into the post-cut equilibrium without re-solving the full system.

\begin{lemma}[Marginal extension of equilibrium]~\label{lemma:equilibrium-marginal-extension}
    Let $\pvec = \{ p_1, \dots, p_m \}$ be a symmetric Nash equilibrium of the game $\game(\tvec, \pmass)$ at utility level $\lambda$, where $\tvec = \{\tx_i\}_{i \in [m]}$. Let $\tx_{m+1}$ be a new transaction with admissible function $f_{m+1}$ of parameters $\lambdalow_{m+1}, \lambdahigh_{m+1}$, and let $\fclamp_{m+1}$ be its associated clamped function. Define $ \tvec' := \tvec \cup \{\tx_{m+1}\} $ and $ \pmass' := \pmass + \fclamp_{m+1}(\lambda)$.
    Then $\pvec' := \pvec \cup \{ \fclamp_{m+1}(\lambda) \} $ is the unique symmetric mixed Nash equilibrium of $\game(\tvec', \pmass')$, supported at the same utility level $\lambda$. In particular, $p_i$ is unchanged 
    for every $i \in [m]$, and $\tx_{m+1}$ attains $p_{m+1} = \fclamp_{m+1}(\lambda)$.
\end{lemma}

\begin{proof}[Proof of~\cref{lemma:equilibrium-marginal-extension}]
    We verify that the profile $(p_1, \ldots, p_m,\, \fclamp_{m+1}(\lambda))$ satisfies the equilibrium conditions in $\game(\tvec', \pmass')$ at level $\lambda$.
    For each $i \in [m]$, the identity $p_i = f_i(\lambda)$ holds by hypothesis on the original NE. Since $u_i$ depends only on $p_i$, this identity is a local condition at $\tx_i$ and is unaffected by enlarging the transaction set. For $i = m+1$, the identity $p_{m+1} = f_{m+1}(\lambda)$ holds by construction. The clearing condition is immediate from the original equilibrium:
    \begin{align*}
        \sum_{i = 1}^{m+1} \fclamp_{i}(\lambda) =
        \sum_{i = 1}^{m} \fclamp_{i}(\lambda) +  \fclamp_{m+1}(\lambda)  = 
        \pmass + \fclamp_{m+1}(\lambda) = 
        \pmass'.
    \end{align*}
    Hence $\pvec'$ is a symmetric mixed NE of $\game(\tvec', \pmass')$ at level $\lambda$. By \autoref{thm:uniqueness}, this is the unique equilibrium of $\game(\tvec', \pmass')$.
    
\end{proof}

Our argument for the correctness of the cut-and-paste method has three pieces: (1)~characterize how a bribe $b$ shifts the inclusion probability of $\tx_c$; (2)~apply~\cref{lemma:equilibrium-marginal-extension} to add $\tx_c$ back into the 
post-cut equilibrium; (3)~solve for the bribe $b^*$ that delivers the 
target probability.

\paragraph{Inclusion under bribe} 
If $\adv$ offers bribe $b$ to all parties to 
censor $\tx_c$, the expected utility from including $\tx_c$ becomes
\begin{align*}
  u_c'(p_c) \;:=\; u_c(p_c) - b.
\end{align*}
In direct analogy with $f_c$, define $f_c'(\,\cdot\,, b)$ as the inverse of 
$u_c'$ in its first argument: $f_c'(\lambda, b)$ is the unique value 
satisfying $u_c'\bigl(f_c'(\lambda, b)\bigr) = \lambda$ whenever 
$f_c'(\lambda, b) \in [0, 1]$. A quick rearrangement yields the 
\emph{bribe-shift identity}:
\begin{align*}
  f_c'(\lambda, b) \;=\; f_c(\lambda + b),
\end{align*}
since $u_c\bigl(f_c'(\lambda, b)\bigr) - b = \lambda$ rearranges to 
$u_c\bigl(f_c'(\lambda, b)\bigr) = \lambda + b$, and applying $f_c$ to both sides gives the identity. The shifted demand $f_c'(\,\cdot\,, b)$ is 
admissible (\cref{def:admissible}) with parameters 
$(\lambdalow_c - b,\, \lambdahigh_c - b)$, since horizontal translation 
preserves admissibility.

\paragraph{Adding $\tx_c$ back} 
We apply 
\cref{lemma:equilibrium-marginal-extension} to the post-cut equilibrium 
$(\pvec_{-c}, \lambda_{-c})$, treating $\tx_c$ -- equipped with demand 
$f_c'(\,\cdot\,, b)$ -- as the new transaction. The lemma gives the unique 
equilibrium of the augmented game
\begin{align*}
  \game \bigl( \tvec,\; (\pmass - p_c^*) + \fclamp_c'(\lambda_{-c}, b) \bigr):
\end{align*}
the utility level remains $\lambda_{-c}$, every $\tx_j \neq \tx_c$ retains 
its probability from $\pvec_{-c}$, and $\tx_c$ attains probability 
$\fclamp_c'(\lambda_{-c}, b)$.

\paragraph{Solving for $b^*$} For this augmented equilibrium to coincide with 
the equilibrium of $\game(\tvec, \pmass)$ under bribe $b$ at target 
probability $p_c^*$, two conditions must hold: (1)~the augmented capacity 
$(\pmass - p_c^*) + \fclamp_c'(\lambda_{-c}, b)$ must equal $\pmass$; and 
(2)~the probability $\fclamp_c'(\lambda_{-c}, b)$ assigned to $\tx_c$ must 
equal $p_c^*$. Both reduce to the same equation,
\begin{align*}
  \fclamp_c'(\lambda_{-c}, b^*) \;=\; p_c^*.
\end{align*}
Applying the bribe-shift identity, the equation becomes 
$\fclamp_c(\lambda_{-c} + b^*) = p_c^*$, yielding for $p_c^* \in (0, 1)$ the 
closed-form bribe
\begin{align*}
  \boxed{\;b^* \;=\; u_c(p_c^*) \;-\; \lambda_{-c}.\;}
\end{align*}
The boundary case $p_c^* = 0$ (full censorship, $\censorprob = 1$) is 
handled by continuity, giving $b^* = \lambdahigh_c - \lambda_{-c}$.

Finally, recall from \cref{def:bribery-cost} that the bribery cost is the adversary's \emph{expected total expenditure}: the per-proposer bribe $b^*$ is collected by each proposer with probability $1 - p_c^* = \censorprob^{1/n}$, so the bribery cost equals $n\, \censorprob^{1/n}\, b^*$. All eCR values we report (\cref{sec:evaluation}) use this total.

\subsection{Applying Cut-and-Paste to Concrete TFMs}

We now apply the cut-and-paste recipe to all three TFMs.

\paragraph{Duplication-Penalizing TFM}
For this TFM, recall from \Cref{ex:duplication} that
\begin{align*}
    u_i(p_i) &= (1-p_i)^{n-1} \cdot t_i \quad
    \quad f_i(\lambda) = 1 - \bigg( \frac{\lambda}{t_i} \bigg)^{\frac{1}{n-1}}.
\end{align*}

Since $f_i$ is admissible, we can use $b =  u_i(p_i) - \lambda$, which gives
\begin{align*}
    b = (1-p_i)^{n-1} \cdot t_i - \lambda \quad \implies
    b = s^{\frac{n-1}{n}} \cdot t_i - \lambda.
\end{align*}

\paragraph{Even-Split TFM}
For this TFM, recall from \Cref{ex:even-split} that
\begin{align*}
    u_i(p_i) &= t_i \cdot \frac{1-(1-p_i)^{n}}{n\,p_i},
\end{align*}
whose inverse $f_i$ is admissible with parameters $(\underline{\lambda}_i, \overline{\lambda}_i) = (t_i/n,\, t_i)$, though it has no closed form. The bribe, however, does: substituting $p_c^* = 1 - s^{1/n}$ into $b = u_c(p_c^*) - \lambda_{-c}$ and using $(1-p_c^*)^n = s$ gives
\begin{align*}
    b = t_c \cdot \frac{1-(1-p_c^*)^{n}}{n\,p_c^*} - \lambda_{-c} \quad \implies
    b = \frac{t_c\,(1-s)}{n\,(1 - s^{\frac{1}{n}})} - \lambda_{-c}.
\end{align*}
As a sanity check, as $n \to \infty$ we have $n(1-s^{1/n}) \to \ln(1/s)$, so the per-proposer bribe approaches $t_c(1-s)/\ln(1/s) - \lambda_{-c}$: even with many proposers, each proposer's marginal stake in $\tx_c$ remains a constant fraction of $t_c$, unlike the collective TFM below where it vanishes as $1/n$.

\paragraph{Collective TFM.}
For this TFM, recall from \Cref{ex:socialist} that
\begin{align*}
    u_i(p_i) &= (1-p_i)^{n-1} \cdot \frac{t_i}{n} \quad \text{and} 
    \quad f_i(\lambda) = 1 - \bigg( \frac{\lambda}{\frac{t_i}{n}} \bigg)^{\frac{1}{n-1}}.
\end{align*}

Since $f_i$ is admissible, we can use $b =  u_i(p_i) - \lambda$, which gives
\begin{align*}
    b = (1-p_i)^{n-1} \cdot \frac{t_i}{n} - \lambda \quad \implies
    b = s^{\frac{n-1}{n}} \cdot \frac{t_i}{n} - \lambda.
\end{align*}

Note that this reduces to the duplication-penalizing TFM, with just a $1/n$ factor applied to all tips. This reflects the fact that the mechanisms are game-theoretically quite similar. For both mechanisms, the only factor relevant to the validator's utility is whether anyone else has included the transaction. If so, including it in their sub-block doesn't increase their utility; if not, their utility increases by $t_i$ or $t_i/n$.

\section{Comparing TFMs on Idealized and Real-World Data}
\label{sec:evaluation}
In this section, we run Monte-Carlo simulations using~\cref{alg:binsearch} and the cut-and-paste recipe to plot the approximate throughput and eCR as the number of proposers increases. We show how these values vary with congestion, burn, and the adversary's desired censorship probability. 
We compare across the three TFMs, finding that the duplication-penalizing TFM performs best across the board. 

\subsection{Comparing TFMs}
We first describe our baseline parameters and simulation method. We simulate a mempool of 250 transactions for a block that can hold any 200 of these. The bids of the transactions are drawn iid from $\textsc{Pareto}(1, 2.5)$, to simulate the standard environment of many low bidders and a few high bidders. We perform 1000 trials, each with different mempools, and for each trial, we use our ``cut-and-paste'' method to compute its economic censorship resistance and throughput, assuming that an adversary wants to achieve a $90\%$ censorship probability. We graph the average of the trials as the number of validators, $n$, varies. 

\autoref{fig:aft_baseline} compares the three TFMs using the baseline parameters, showing that the duplication-penalizing TFM achieves the highest eCR, and ties the collective TFM for the highest throughput. Recall from \Cref{sec:deriving-censorship-resistance} that the \emph{collective} TFM reduces to the duplication-penalizing TFM with all tips scaled by $1/n$; since a uniform rescaling of the tips merely rescales $\lambda^*$, the two TFMs induce identical equilibrium probabilities, and thus identical throughput. They differ starkly in eCR, however: under the collective TFM each proposer's marginal stake in any single transaction is only $t_i/n$, so the adversary's bribes shrink as $n$ grows, while users continue to pay their full tips. This is why the collective TFM's eCR remains near zero regardless of $n$.

The even-split TFM has lower throughput than the other two: including an already-included transaction still pays $t_i/k$, so duplication is penalized only weakly, and the equilibrium concentrates more probability mass on high-tip transactions. Its eCR sits between the other two TFMs. Compared to the duplication-penalizing TFM, the gap is driven by the price to users: the even-split TFM always charges users their tip, while the duplication-penalizing TFM only does so if their transaction is singly-included. As a result, the duplication-penalizing TFM has significantly higher eCR, as the bribes required are much larger relative to the expected user payment.

\begin{figure}
    \centering
    \includegraphics[width=0.75\linewidth]{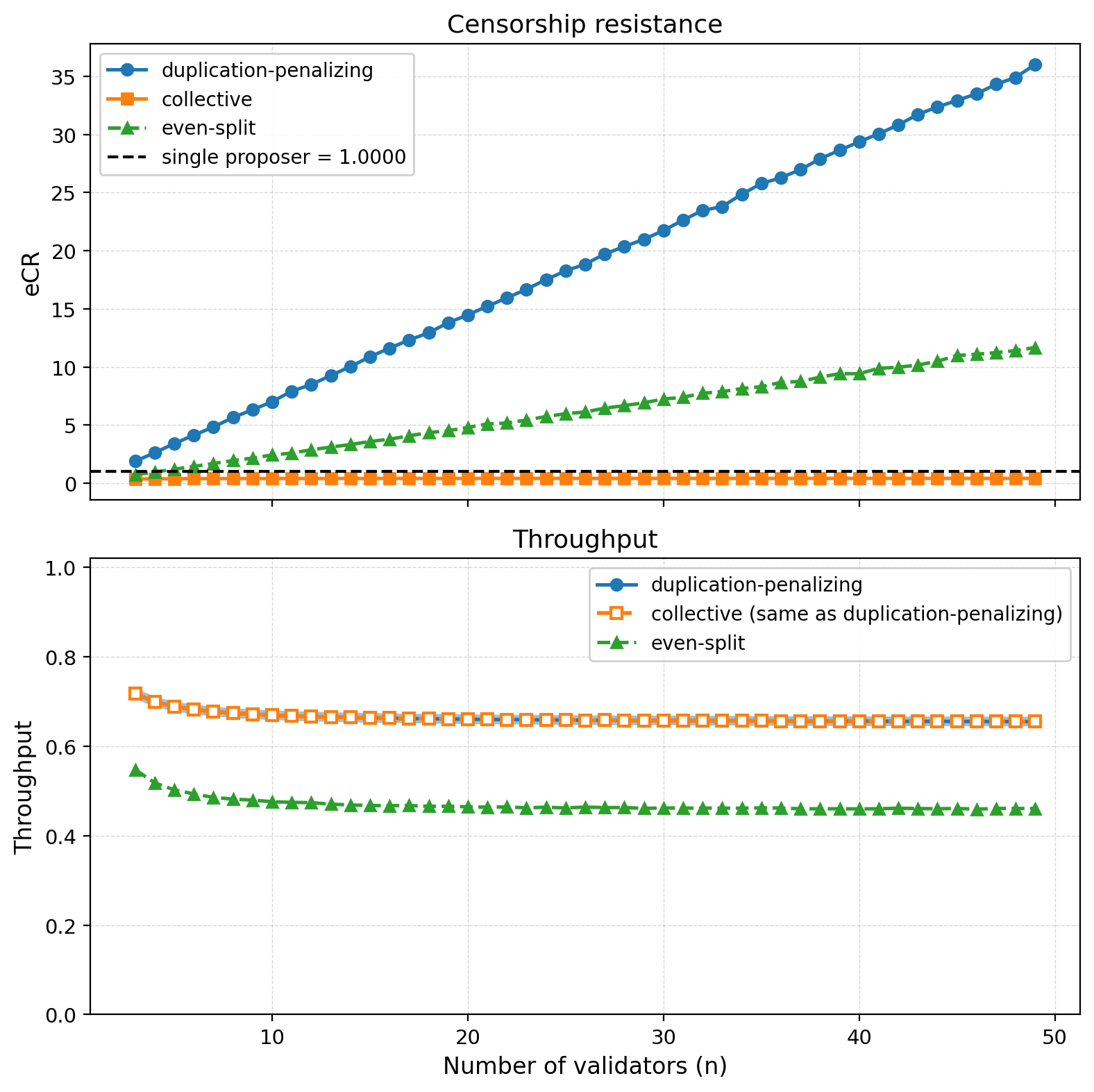}
    \caption{Comparison of three TFMs, with the baseline parameters of $m = 250$, $B=200$, $s=0.9$, and bids drawn iid from $\textsc{Pareto}(1, 2.5)$.}
    \label{fig:aft_baseline}
\end{figure}

\paragraph{Throughput.}
Note that the throughput quickly converges in both cases. For context, a single proposer trivially achieves throughput $1$ by deduplicating its own block; moving to MCP thus incurs a one-time duplication cost (roughly $0.65$ at our baseline under the duplication-penalizing TFM). The key observation is that this cost does not compound: it stabilizes almost immediately, so increasing $n$ further buys censorship resistance at essentially no additional throughput cost. To contextualize the throughput, we compare these algorithms to an optimistic baseline. Suppose that the validators had the explicit goal of maximizing throughput, but were not able to coordinate. It's easy to see that the optimal strategy would be for all validators to choose each transaction uniformly at random. Now, assuming that $m \ge B$, i.e., there are more transactions than block space (so that $\min(B, m) = B$ in \cref{def:throughput-fixed}), the throughput can be written as:
\begin{align*}
    \frac{1}{B}\sum_{i \le m} 1-(1-p_i)^n \approx \frac{1}{B}\sum_{i \le m} 1-e^{-np_i}
\end{align*}
The approximation holds well as long as $p_i$ is small, which happens when $n$ is moderate, as that decreases the sub-block size. Let $y_i = np_i$, the expected number of copies of transaction $i$. Under the uniform distribution, each transaction is chosen with $p_i=\ell/m$, where $\ell=B/n$ is the sub-block size. Thus, $y_i=np_i =B/m$, and the throughput of the uniform distribution is approximately $\frac{m}{B}(1-e^{-B/m})$. This represents an \textit{upper bound} on throughput---profit incentives will decrease the actual throughput, as validators are more likely to include more valuable transactions.

In \autoref{fig:aft_uniform_comparison}, we plot the throughput of the duplication-penalizing and even-split TFMs, as well as their ``gap from uniform,'' in terms of their $\sum_{i \le m} |y_i-B/m|$. This represents their aggregate excess duplication, where ``excess'' is defined relative to the uniform distribution. The main takeaway is that the duplication-penalizing TFM maintains a throughput very close to the optimistic upper-bound of a uniform distribution.
\begin{figure}
    \centering
    \includegraphics[width=0.7\linewidth]{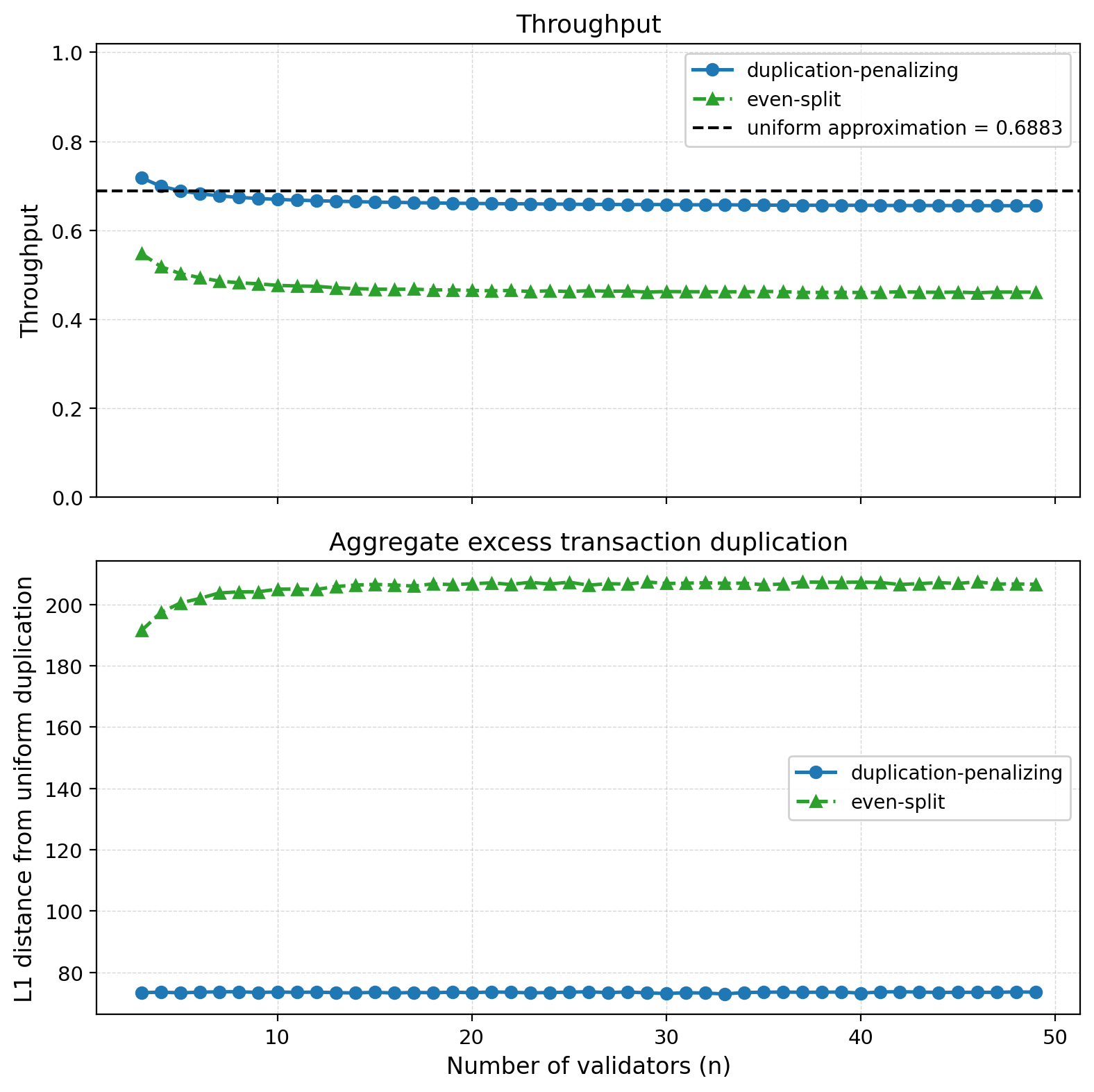}
    \caption{Comparing throughput to the optimistic upper bound of a uniform distribution. We use the baseline parameters of $m = 250$, $B=200$, $s=0.9$, and bids drawn iid from $\textsc{Pareto}(1, 2.5)$.}
    \label{fig:aft_uniform_comparison}
\end{figure}

\paragraph{Burn.}
So far, we have been assuming no burn; users pay only the payment from the TFM. Burn hurts eCR significantly, since higher burn means that adversaries can bribe proposers for much less than the users pay. 
To determine a reasonable amount of burn, we use data from Ethereum in January 2026, and find that the base fee was 8.7x the average tip. \autoref{fig:burn} plots the eCR of the three TFMs assuming a burn multiple of 4 and 8.7. All users are charged the same amount of burn.

\begin{figure}
\centering
    \begin{subfigure}{0.67\linewidth}
        \centering
        \includegraphics[width=\linewidth]{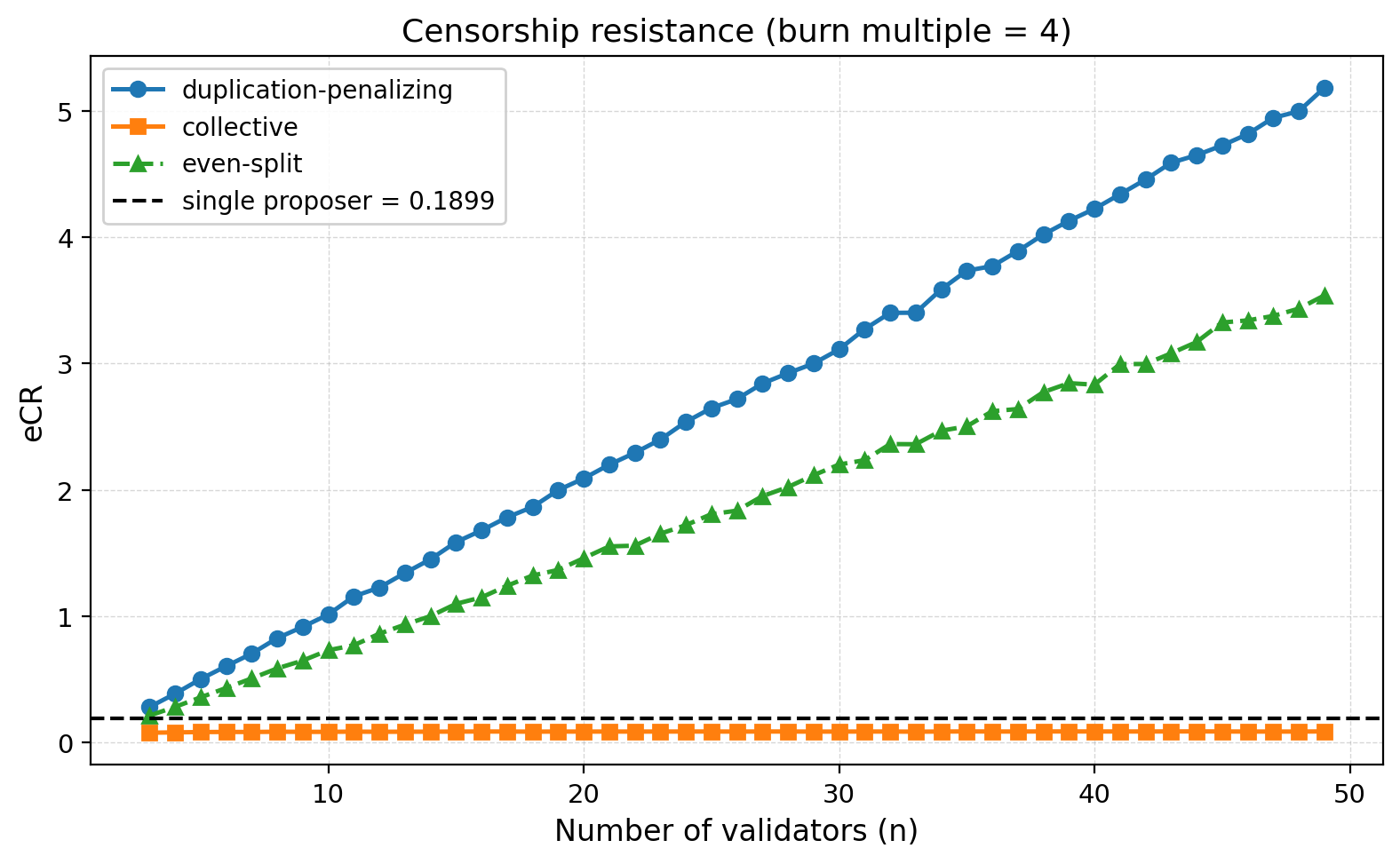}
        \caption{A burn multiple of 4}
        \label{fig:burn_a}
    \end{subfigure}
    \begin{subfigure}{0.67\linewidth}
        \centering
        \includegraphics[width=\linewidth]{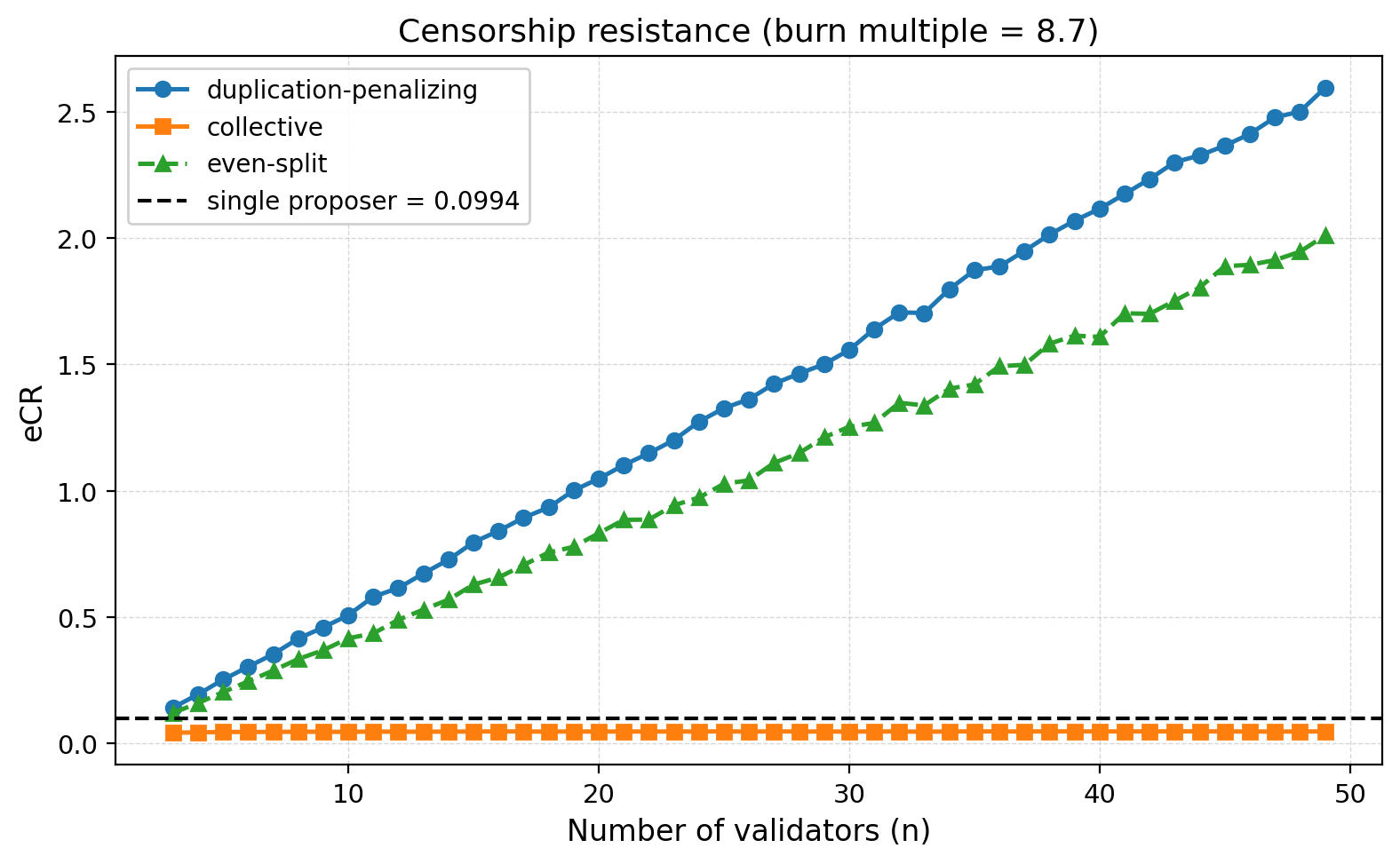}
        \caption{A burn multiple of 8.7}
        \label{fig:burn_b}
    \end{subfigure}
    \caption{Comparing eCR assuming that the burn is either 4x or 8.7x the average tip. We use the baseline parameters of $m = 250, B = 200, s=0.9$, and the bids are drawn iid from $\textsc{Pareto}(1, 2.5)$.}
    \label{fig:burn}

\end{figure}

Though burn significantly hampers the eCR, we see that eCR still increases roughly linearly with $n$, and so increasing the number of validators can still result in high eCR. For example, with a burn multiple of 8.7 and 49 validators, the eCR is approximately 2.5, so an adversary would need to pay 2.5x more than the user to censor the average transaction. The figure also plots the single-proposer baseline, which suffers greatly as burn increases. At a burn multiple of 8.7, the single-proposer eCR is approximately 0.09, so an adversary can censor the average transaction by paying only 9\% of what the user pays. 

Further, note that increasing the burn brings the duplication-penalizing and even-split TFMs closer together. With no burn and 49 validators, the eCR was approximately 35 and 10, for the duplication-penalizing and even-split TFMs respectively. At a burn multiple of 8.7, the eCR was approximately 2.5 and 1.8. This is because the burn is a flat fee charged to all transactions, and so the differences in user payments under the two TFMs are diluted. 

\subsection{Congestion and Adversary Incentives}
Here, we evaluate the impact of congestion and the adversary's goal on the censorship resistance. We only consider the duplication-penalizing TFM, as it beats the other two in both eCR and throughput.

For congestion, with the baseline parameters, we have 250 transactions with a block size of 200. In \autoref{fig:congestion}, we plot the eCR for both more congested and less congested cases. We fix the block size at 200, and vary the number of transactions. As congestion increases, the eCR goes down, because the opportunity cost of censoring a transaction decreases (it's more likely that there is another good transaction to take its place). 

\begin{figure}[t]
    \centering
    \includegraphics[width=0.7\linewidth]{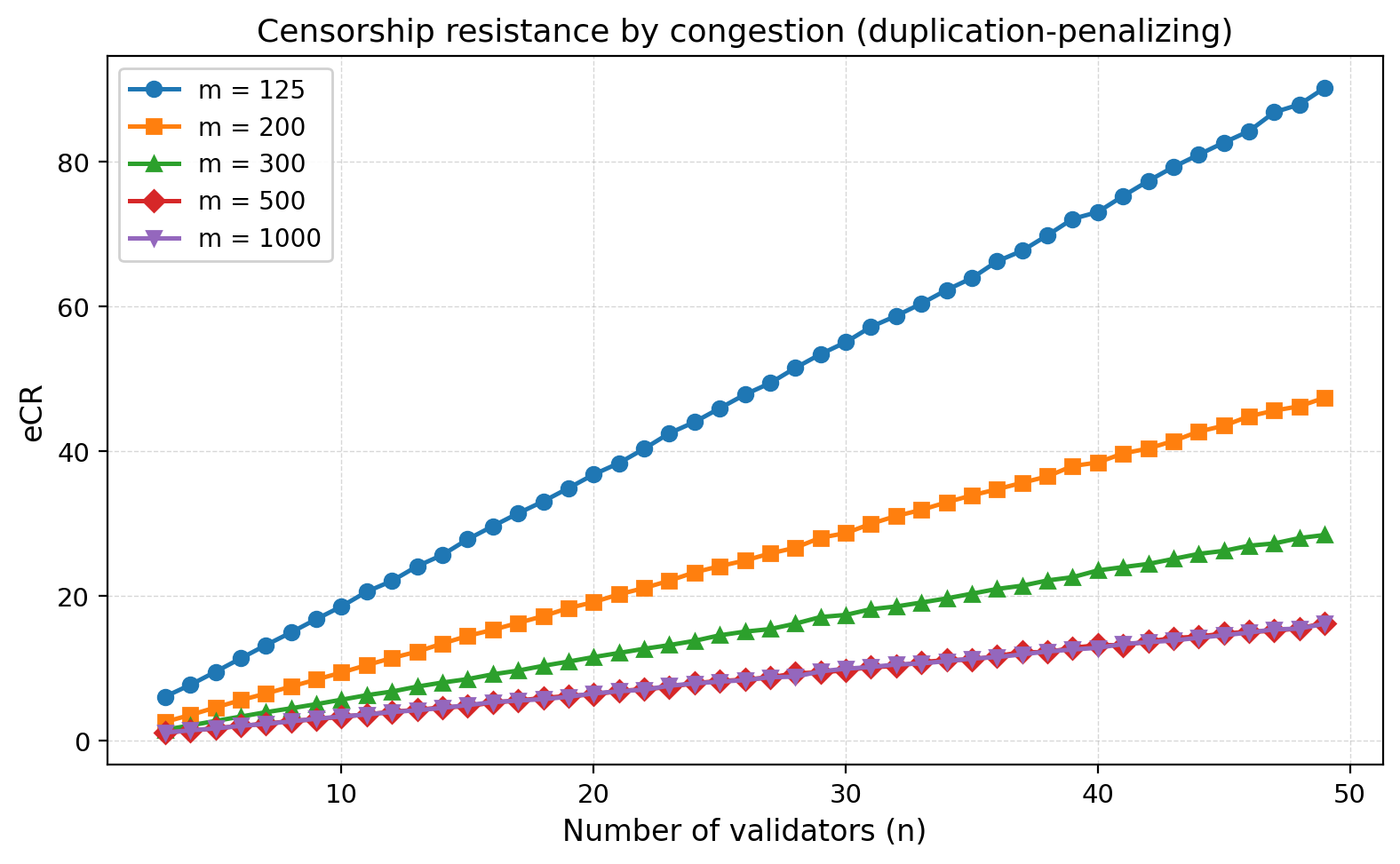}
    \caption{Comparing eCR as congestion varies, using the baseline parameters of $B = 200, s=0.9$, and bids drawn iid from $\textsc{Pareto}(1, 2.5)$. The number of transactions $m$ varies in the different curves.}
    \label{fig:congestion}
\end{figure}

In \autoref{fig:adv_succ}, we plot the eCR for a range of adversary censorship probabilities. Our baseline assumes that the adversary wants to censor a transaction with 90\% probability; we plot for both more lenient and more strict adversaries. As the adversary's desired censorship probability decreases, the eCR also decreases, as smaller bribes are required.

\begin{figure}[t]
    \centering
    \includegraphics[width=0.7\linewidth]{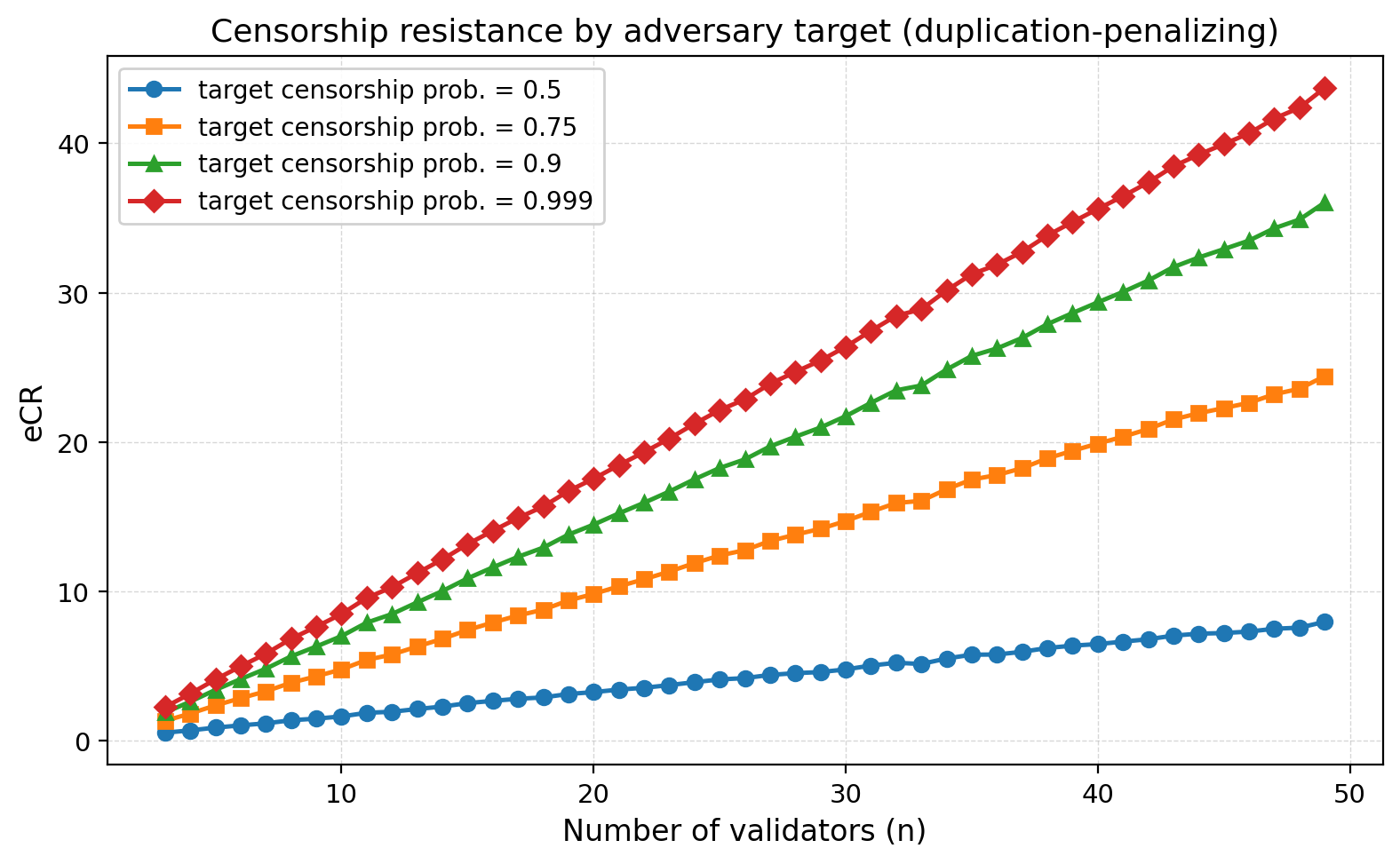}
    \caption{Comparing eCR as the adversary's censorship probability varies, using the baseline parameters of $B = 200, m = 250$, and bids drawn iid from $\textsc{Pareto}(1, 2.5)$. The adversary's censorship probability varies in the different curves.}
    \label{fig:adv_succ}
\end{figure}

\subsection{Analysis on Empirical Data}

\newcommand{\Block}{\mathcal{B}}
\newcommand{\eCR}{\mathsf{eCR}}

We next perform the comparison on real Ethereum transaction data. We use a window of $10{,}000$ recent Ethereum blocks, covering block numbers $24{,}341{,}338$ through $24{,}351{,}341$. For each block $\Block$, we compute each transaction's tip as $\texttt{gasPrice} - \texttt{baseFeePerGas}$ and keep only transactions with strictly positive tip, since these are the transactions that are afforded any eCR. This yields a block-specific set of valid transactions, whose size we denote by $m_{\Block}$.

For each block $\Block$, we fix the target total block size at $B_{\Block} = 0.8\, m_{\Block}$, so that the protocol aims to include roughly $80\%$ of the positive-tip transactions in that block. We then vary the number of validators $n$ from $2$ to $50$, and for each value of $n$ set the per-validator probability-mass budget to $\ell_{\Block} = B_{\Block} / n$, so that total capacity $n\ell_{\Block}$ remains fixed as $n$ changes. We run this procedure for the three TFMs in our study: duplication-penalizing, collective, and even-split.

Given the positive-tip transactions of a block $\Block$ and a choice of $n$, we first solve for the baseline single-proposer equilibrium probabilities $(p_i)_i$ under the relevant TFM. From these probabilities we compute throughput by taking the expected number of distinct transactions included by at least one of the $n$ validators and normalizing by the target block size. To measure economic censorship resistance, we take a sample of $20$ target transactions per block. For each sampled target transaction $\tx_i$, we use the cut-and-paste method to compute the per-validator bribe needed to achieve a target \emph{total} censorship probability of $90\%$. This gives a post-bribe target inclusion probability $p_i'$ and a per-censor bribe $b_i$, from which we compute
\[
\eCR_i = \frac{n(1-p_i')\, b_i}{P_i},
\]
where $P_i$ is the baseline expected user payment conditional on inclusion, in accordance with \cref{def:bribery-cost,def:ecnomic-cr-per-tx}.

Finally, for each block $\Block$ and each $n$, we average $\eCR_i$ over the sampled target transactions in that block, and separately record the corresponding throughput. We then average these block-level quantities over the full $10{,}000$-block dataset to obtain the curves shown in \autoref{fig:aft_eth_data_baseline}.

One caveat is in order: this dataset conditions on inclusion. Transactions that were successfully censored never appear on-chain, so the empirical tip distribution may understate the mass of low-tip or contested transactions. We therefore view these experiments as evaluating the TFMs on realistic bid shapes, rather than as a measurement of censorship itself.

\begin{figure}
    \centering
    \includegraphics[width=0.7\linewidth]{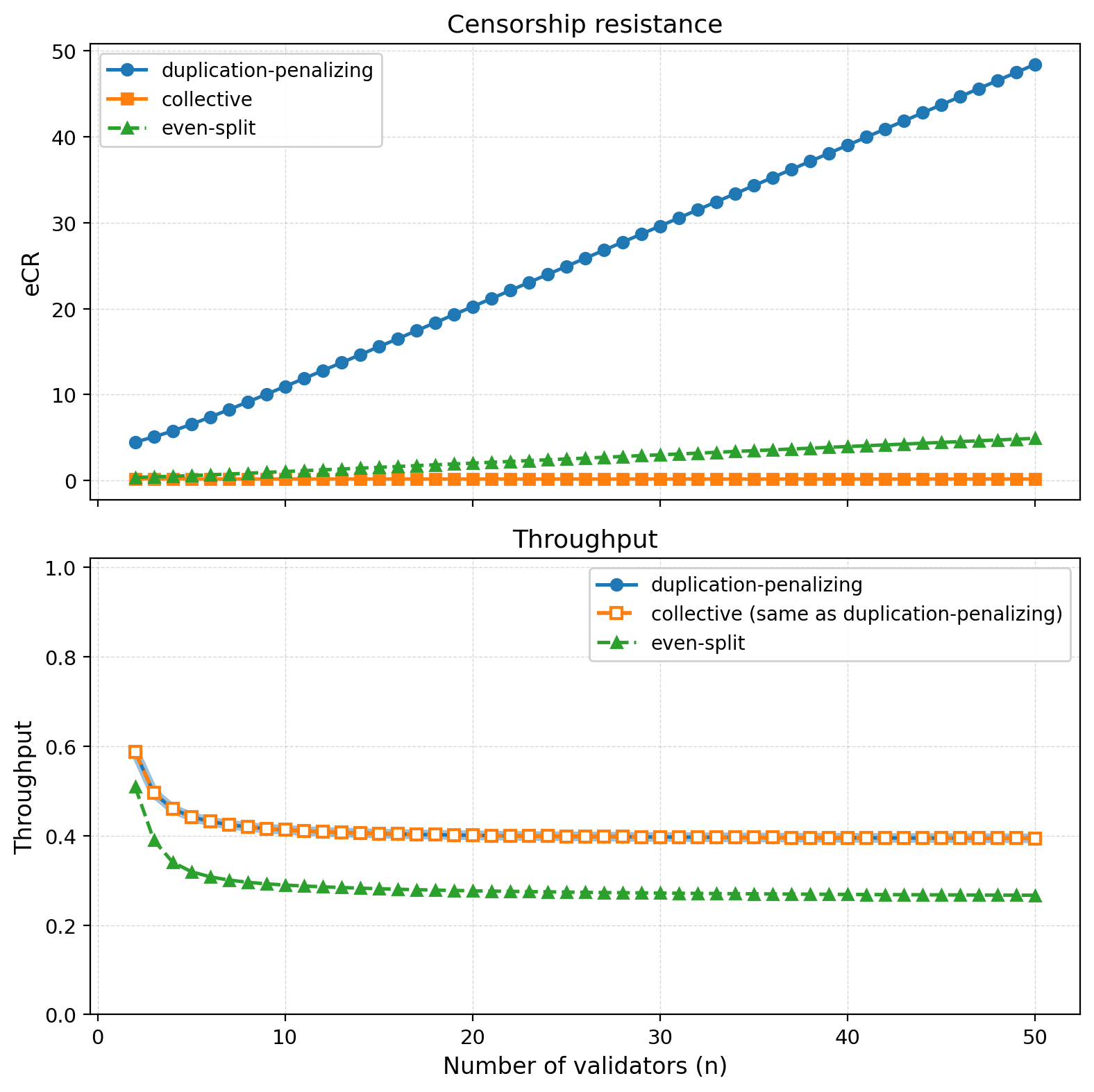}
    \caption{Comparison of three TFMs on recent Ethereum block data, with \(B = 0.8m\), \(s = 0.9\), and bids given by observed positive transaction tips.}
    \label{fig:aft_eth_data_baseline}
\end{figure}

\paragraph{Takeaways from empirical analysis.}
Similar to the analysis on the synthetic data, the analysis on the empirical data shows that the duplication-penalizing TFM dominates the other two in both throughput and economic censorship resistance. The main difference between the empirical data and our synthetic data is that the empirical tip distribution is more skewed; a larger percentage of the tips are closer to 0 in reality. This reduces the throughput, because the probability mass concentrates on the few high value tips. Lower throughput means more duplications, and so the duplication-penalizing TFM has much higher censorship resistance, since the average user payment is low. The link to reproduce any empirical result can be found at \href{https://zenodo.org/records/21831640}{zenodo record number 21831640}.

\section{Conclusion}
\label{sec:conclusion}
We proposed a novel definition of economic censorship resistance (eCR), and modeled the incentives of rational proposers in the presence of a censoring adversary. We found an efficient algorithm to determine the equilibrium of our game, and our algorithm works for a variety of TFMs. We used this algorithm to compare the eCR and throughput of three TFMs on real and synthetic data. Overall, we found that:
\begin{enumerate}
    \item MCP systems enjoy much higher eCR than single proposer systems
    \item Burnt user fees can significantly lower eCR
    \item eCR scales linearly with the number of validators, and thus strong eCR can be achieved even in the presence of significant burn
    \item Of the TFMs we study, the duplication-penalizing TFM dominates in both throughput and eCR, across a wide range of parameters
\end{enumerate}

We suggest three avenues for future work. The first is to analyze a wider range of TFMs. In particular, one may be able to prove that the duplication-penalizing TFM is optimal among a family of TFMs, or to prove some theoretical upper bounds on censorship, as we informally did for throughput. A user-side equilibrium analysis---endogenous tips, and robustness to user--proposer collusion under the duplication-penalizing TFM (\cref{rem:user-side})---is another natural target. This would complement our empirical analysis, and provide guidance on how to best implement an MCP system.

Second, one could analyze different bribery attacks. We consider a simple setting, where the censoring-adversary must declare a transaction to be censored, along with a bribe for any proposer who doesn't include that transaction. But censoring behavior may take other forms. For example, the bribe could be paid out only if \textit{no} proposer includes that transaction. 

Finally, we assume that all proposers have access to a shared mempool. But some proposers may have private access to some transactions. The simplest example of this is private order flow, but this may also occur for geographical reasons. For example, proposers located far from the rest of the network may receive transactions from their region before the rest of the network, and thus have semi-private access in some timing windows. Our model, and the duplication-penalizing TFM in particular, heavily rewards proposers with private transactions. Thus, MCP systems may serve as a counterbalance to the latency pressures that encourage geographical co-location.

\bibliography{refs}

\end{document}